\documentclass[a4paper,UKenglish,cleveref, autoref, thm-restate, numberwithinsect]{lipics-v2021}

\usepackage{amsthm}
\usepackage{tikz,tikz-qtree}
\usetikzlibrary{decorations.pathmorphing}
\usetikzlibrary{backgrounds,calc,positioning}
\usepackage{complexity,bm}
\usepackage{todonotes}

\usepackage{algorithm}
\usepackage{algpseudocode}

\nolinenumbers

\colorlet{darkgreen}{green!80!black}

\newclass{\EXPTIME}{EXPTIME}
\newclass{\ALOGTIME}{ALOGTIME}
\newclass{\DLOGTIME}{DLOGTIME}
\newclass{\LogDCFL}{LogDCFL}
\newclass{\LogCFL}{LogCFL}
\newclass{\uNC}{uNC}
\newclass{\uTC}{uTC}
\newclass{\uAC}{uAC}
\newclass{\coUL}{coUL}

\newcommand{\num}{\mathsf{num}}
\newcommand{\derive}{\Rightarrow}
\newcommand{\Bin}{\mathsf{Bin}}
\newcommand{\val}{\mathsf{val}}
\newcommand{\height}{\mathsf{height}}

\newcommand{\hs}{\mathsf{st}}
\newcommand{\unfold}{\mathsf{unfold}}
\newcommand{\rest}{\mathord\restriction}
\newcommand{\CNF}{\mathsf{CNF}}
\newcommand{\rCNF}{\mathsf{pCNF}}
\newcommand{\acCNF}{\mathsf{acCNF}}
\newcommand{\CNFac}{\acCNF^\ge}
\newcommand{\CNFack}{\acCNF^{\ge k}}
\newcommand{\LS}{\mathsf{L}}
\newcommand{\StN}{\mathsf{St}^\ge}
\newcommand{\StNk}{\mathsf{St}^{\ge k}}
\newcommand{\StNterm}{\StN_{\mathsf{term}}}
\newcommand{\StNpointer}{\StN_{\mathsf{pointer}}}
\newcommand{\StNkterm}{\StNk_{\mathsf{term}}}
\newcommand{\StNkpointer}{\StNk_{\mathsf{pointer}}}
\newcommand{\StNdag}{\StN_{\mathsf{dag}}}
\newcommand{\StNtslp}{\StN_{\mathsf{tslp}}}
\newcommand{\StNkdag}{\StNk_{\mathsf{dag}}}
\newcommand{\StNktslp}{\StNk_{\mathsf{tslp}}}

\hideLIPIcs

\title{On the complexity of computing Strahler numbers}

\author{Moses Ganardi}{University of Kaiserslautern-Landau (RPTU), Germany}{moses.ganardi@cs.rptu.de}{https://orcid.org/0000-0002-0775-7781}{}
\author{Markus Lohrey}{Universit\"at Siegen, Germany}{lohrey@eti.uni-siegen.de}{https://orcid.org/0000-0002-4680-7198}{}

\authorrunning{M.~Ganardi, M.\ Lohrey}

\Copyright{Moses Ganardi, Markus Lohrey}

\begin{document}

\maketitle
\begin{abstract}
It is shown that the problem of computing the Strahler number of a binary tree given as a term
 is complete for the circuit complexity
class uniform $\NC^1$. For several variants, where the binary tree is given by a pointer structure or in a succinct form by
a directed acyclic graph or a tree straight-line program, the complexity of computing the Strahler number is determined as well.
We show that the problem of deciding whether a given context-free grammar in Chomsky normal form produces a derivation tree with a Strahler number of at least $k$ is $\P$-complete. If the derivation tree is restricted to be acyclic, the problem becomes $\PSPACE$-complete.

\end{abstract}

\section{Introduction}

\subparagraph{Strahler numbers.}
The main topic of this paper is the complexity of computing \emph{Strahler numbers} of binary trees.
The \emph{Strahler number} of a binary tree $t$ is a parameter $\hs(t)$ that can be defined recursively as follows:
\begin{itemize}
\item If $t$ consists of a single node then $\hs(t) = 0$.
\item If the root of $t$ has the left (resp., right) subtree $t_1$ (resp., $t_2$) then
\begin{equation} \label{eq-strahler-operation}
	\hs(t) = \begin{cases} \hs(t_1)+1 & \text{if } \hs(t_1) = \hs(t_2), \\ \max\{\hs(t_1), \hs(t_2)\} & \text{if } \hs(t_1) \neq \hs(t_2). \end{cases}
\end{equation}
\end{itemize}
The Strahler number, also known as the \emph{Horton-Strahler number}, first appeared in the area of hydrology, where Horton used it in a paper from 1945 \cite{Ho45}
to define the order of a river. The correspondence to binary trees comes from the fact that a system of joining rivers
can be viewed as a binary tree (unless there are bifurcations, where a river splits into two streams).
In 1952, Strahler \cite{Stra52} (also a hydrologist) further developed Horton's ideas. 

There are numerous applications of Strahler numbers in computer science, where they appeared also under different names (e.g., register function, tree dimension).
Ershov \cite{Ers58} showed that the minimal number of registers needed to evaluate an arithmetic expression is exactly the Strahler number of the 
syntax tree of the arithmetic expression. The Strahler number has also found many applications in formal language theory \cite{ChMo90,EhRoVe81,GiSpa68,SOS23}. For context-free grammars, the relation comes from the following fact: Let $G$ be a context-free grammar in Chomsky normal form and let $t$ be a derivation tree of $G$.
Then $\hs(t)+1$ is exactly the minimal index among all derivations corresponding to $t$,
where the index of a derivation $S = w_0 \derive_G w_1 \derive_G w_2 \derive_G \cdots \derive_G w_n$ 
 is the maximal number of nonterminals in one of the $w_i$~\cite{GiSpa68}. 
Finite-index context-free grammars~\cite{AtigG11}, i.e., grammars where every produced word has a derivation of bounded index,
play an important role in the recent decidability proof of the reachability problem in one-dimensional pushdown VASS~\cite{BiziereC25}.
Strahler numbers have been also investigated in the context of Newton iteration \cite{EKL10,PiSaSo12}, parity games \cite{DaJuTh20}, and 
social networks \cite{ADDGG04}.
The distribution of the  Strahler number of a random tree has been studied by several authors \cite{DevKru95,DeKru96,FRV79,Kemp79,Kru99}.
For more information on Strahler numbers and their applications in computer science, the reader may consult
the surveys \cite{EsparzaLS14,Vie90}.

The above-mentioned applications naturally lead to the question for the precise complexity of computing the Strahler number
of a given tree. This problem has a straightforward linear time algorithm: simply compute the Strahler number for all subtrees
bottom-up using the definition of the Strahler number.
A straightforward recursive algorithm can be implemented on a deterministic Turing machine, running in logspace and polynomial time and equipped with an auxiliary stack,
which puts the problem in the class $\LogDCFL \subseteq \NC^2 \subseteq \DSPACE(\log^2 n)$~\cite{Sudborough78}. In particular, Strahler numbers can be computed in polylogarithmic time with polynomially many processors.
Alternatively, one can implement a recursive evaluation in $\mathcal{O}(\log n \log \log n)$ space.
To do so, one uses the fact that the Strahler number of a binary tree with $n$ leaves is bounded by $\log_2(n)$,
and hence its bit length is $\mathcal{O}(\log \log n)$.
To the best of our knowledge, the precise complexity of computing Strahler numbers has not yet been established.

\subparagraph{Contributions.}
Our first goal is to pinpoint the precise parallel complexity of computing the Strahler number. 
As explained above, the problem belongs to $\NC$, but for instance the existence
of a logspace algorithm is by no means obvious. Formally, we 
consider the decision problem,
asking whether $\hs(t) \geq k$ for a given tree $t$ and a given number $k$. 
We show that this problem is $\NC^1$-complete
if $t$ is given in term representation.\footnote{For instance, $bbaabaa$ (or $b(b(a,a),b(a,a))$ with brackets)
is the term representation of a complete binary tree of height 2, where $b$ denotes an inner node and $a$ denotes a leaf.}
Recall that $\NC^1$ is the class of all problems that can be solved by a uniform\footnote{All circuit complexity classes refer to their uniform variants in this paper. In Section~\ref{sec-CC} we will say more about uniformity.} family of bounded fan-in circuits of polynomial-size and logarithmic depth,
which is a subclass of deterministic logarithmic space ($\L$ for short).
If the term $t$ is given as a pointer structure, i.e., by an adjacency list or matrix, then checking $\hs(t) \geq k$ 
is complete for deterministic logspace.

As a corollary, given arithmetic expression $e$ in term representation,
one can compute an optimal straight-line code for $e$ in $\NC^1$,
i.e., a sequence of statements $x := y \circ z$ for registers $x,y,z$ and an elementary arithmetic operation $\circ$.
Here, optimal means that the number of used registers is minimal.
For this, one has to compute the Strahler number of every subexpression of $e$; see also \cite[Section~2]{EsparzaLS14}.

Let us give a high-level idea of the $\NC^1$-membership proof.
The first step is to ``balance'' the input tree $t$
by computing a so-called \emph{tree straight-line program} (TSLP) for $t$, whose depth is 
logarithmic in the size of $t$.
This can be done in $\TC^0$ by a result from \cite{GanardiL19}.  Roughly speaking, 
a TSLP is a recursive decomposition of a tree into subtrees
and so-called contexts (subtrees, where a smaller subtree is removed). Originally, TSLPs
were introduced as a formalism for grammar-based tree compression; see~\cite{Lohrey15} for more details.
The next step is to convert the TSLP for $t$ into a bounded fan-in Boolean circuit, that decides whether $\hs(t) \geq k$. 
Furthermore, the polynomial size and logarithmic depth of the TSLP should be preserved to obtain an $\NC^1$ upper bound.
A straightforward construction only leads to such a Boolean circuit with \emph{unbounded} fan-in OR-gates.
We obtain a bounded fan-in circuit by carefully analyzing 
the unary linear term functions computed by contexts when binary nodes are interpreted according to \eqref{eq-strahler-operation}.

For the $\NC^1$-hardness, we show that the Boolean formula
problem, which is one of the best known $\NC^1$-complete problems \cite{Bus87}, 
can be reduced to the problem of computing the Strahler number of
a tree given in term representation. A similar reduction from the monotone circuit value problem shows that the computation of the Strahler number is $\P$-complete
when the input tree is given succinctly by a directed acyclic graph (DAG) or a TSLP. 

We also consider the problem of checking $\hs(t) \geq k$ for a fixed value $k$ that is not part
of the input (the input only consists of the tree $t$). If $t$ is given in term representation (resp., pointer representation) then this problem is $\TC^0$-complete for all $k \geq 4$ (resp.,
$\LS$-complete for all $k \geq 3$). Moreover, if $t$ is given by a DAG, then this problem is $\LS$-complete for all $k \ge 3$, whereas for TSLP-represented trees the problem is $\NL$-complete for all $k \geq 2$.

In Section~\ref{sec-CNF}  we briefly report on some results concerning the maximal Strahler number of derivation trees of a given context-free grammar in Chomsky normal form (CNF). It is known to be undecidable, 
whether every word produced by a given context-free grammar has a derivation tree 
of Strahler number at most a given bound $k$~\cite[Theorem~5]{Gruska71b}. Here, we are interested in the question, whether a given CNF-grammar $G$ produces at least one derivation tree $t$
with $\hs(t) \geq k$ for a given number $k$. We show that this problem is $\P$-complete. 
For the upper  bound we compute
the maximal Strahler number among all derivations trees of a given CNF-grammar (which can be $\infty$) by a fixpoint iteration procedure. 
The lower bound follows from the above-mentioned $\P$-completeness result for directed acyclic graphs.

Finally, we also consider the restriction to \emph{acyclic derivation trees}.
A derivation tree is called acyclic if there is no nonterminal that appears more than once on a path in the derivation tree.
The motivation for this restriction comes from the recent paper \cite{LRZ25}, where it was shown that the intersection non-emptiness
problem for a given list of group DFA\footnote{A group DFA is a deterministic finite automaton, where for every input letter $a$
the $a$-labelled transitions induce a permutation of the set of states.} 
 plus a single context-free grammar is $\PSPACE$-complete. For general DFA, this problem is $\EXPTIME$-complete \cite{SwernofskyW15}.
 Moreover, if the context-free grammar $G$ is such that for some constant $k$, all acyclic derivation trees of $G$ have 
 Strahler number at most $k$, then the intersection problem (with the finite automata restricted to group DFA) is $\NP$-complete.
 In \cite{ LRZ25}, it was shown that the problem whether a given CNF-grammar has an acyclic derivation tree of Strahler number at least $k$ is in $\NP$, when $k$ is a fixed constant. We show that the problem is $\NP$-complete already for $k=2$.
 Finally, when $k$ is part of the input, we show that the problem becomes $\PSPACE$-complete.
  
 \subparagraph{Broader context: tree evaluation and tree balancing.}
 The problem of computing the Strahler number of a given tree is a special instance of a  \emph{tree evaluation problem}:
The input is a rooted tree where each leaf is labelled with a value from a domain $A$,
and each inner node carries a (suitably specified) function $f \colon A^r \to A$ where $r$ is the number of its children.
The goal is to compute the value of the root, obtained by evaluating the functions at each node from bottom to top.
For the case of Strahler numbers we have $A = \mathbb{N}$, every leaf is labelled with $0$ and there 
 is only one binary operation implicitly defined by \eqref{eq-strahler-operation} (or explicitly by  \eqref{eq-strahler-algebra} on page~\pageref{page-st}).
 The corresponding algebra will be called the \emph{Strahler algebra}.

Other prominent examples are the evaluation problems for Boolean formulas such as $(1 \vee 0) \wedge 1$
and arithmetic expressions over the natural numbers (or other rings) such as $(1 + 2) \times (3 + 4)$.
Boolean formula evaluation is $\NC^1$-complete~\cite{Bus87}.
In fact, the acceptance problem of a fixed tree automaton or, equivalently, evaluating an expression over a fixed finite algebra
is known to be in $\NC^1$ for every finite algebra~\cite{GanardiL19,Loh01rta}.
By an algebra, we simply mean a set equipped with a set of finitary operations.
More surprisingly, arithmetic expressions can be evaluated in deterministic logspace~\cite{Ben-OrC92,BCGR92,ChiuDL01}
despite the fact that the value of an expression may have polynomially many bits in terms of the size of the expression.

Any algorithm that performs a bottom-up computation over a tree can
be seen as an instance of tree evaluation,
assuming that the local computation at each node is sufficiently simple.
A classical example is Courcelle's theorem, stating that any 
monadic second-order (MSO) definable graph property $\Phi$ can be checked in linear time
over graphs of bounded tree-width~\cite{Courcelle90}.
The standard proof of Courcelle's theorem compiles the MSO formula $\Phi$ into a tree automaton
$\mathcal{A}_\Phi$ that, given a tree decomposition of a graph $G$, verifies
whether $\Phi$ holds in $G$.
As remarked above, tree automata can be simulated in $\NC^1 \subseteq \L$,
and therefore Courcelle's theorem also holds when linear time is replaced by logspace~\cite{ElberfeldJT10}
(assuming the logspace version of Bodlaender's theorem for computing small-width tree decompositions; see \cite{ElberfeldJT10}).
In fact, \cite{ElberfeldJT10} proves a more powerful \emph{solution histogram} version of Courcelle's theorem,
which, in the end, reduces to evaluating arithmetic expressions. 

Recently, the tree evaluation problem has attracted renewed attention due to a surprising result by Cook and Mertz~\cite{CookM24}.
They presented an algorithm that evaluates a complete binary tree of height $h$, whose inner nodes are labelled with 
binary operations over $\{1, \ldots, k\}$ and whose leaves are labelled with elements from $\{1, \ldots, k\}$,
in space $\mathcal{O}(h \log \log k + \log k)$.
A straightforward evaluation takes $\mathcal{O}(h \log k)$ space.
Since the height $h$ is logarithmic in the total input size $n$, the Cook-Mertz algorithm uses $\mathcal{O}(\log n \log \log n)$ space, which comes very close to $\mathcal{O}(\log n)$ space.
It is also a key ingredient in Ryan Williams' recent proof that
any $t$-time bounded Turing machine can be simulated in $\mathcal{O}(\sqrt{t \log t})$ space~\cite{Williams25}.
Notice that the Cook-Mertz algorithm does not give any nontrivial space bounds for the computation of the Strahler number
of a tree $t$, since the height of $t$ can be linearly large in its size.

A standard strategy to evaluate a tree $t$ of size $n$ using small space or in parallel polylogarithmic time is to first \emph{balance} $t$,
i.e., to transform it into an equivalent tree of  depth $\mathcal{O}(\log n)$ and size $\mathsf{poly}(n)$.
In a second step, the reduced depth can often be exploited to evaluate the tree in parallel or in small space.
For example, to evaluate a balanced arithmetic expression in logspace, one can use a result by Ben-Or and Cleve~\cite{Ben-OrC92} that transforms
an arithmetic expression of depth $d$ into a product of $4^d$ many
$(3 \times 3)$-matrices such that the value of the arithmetic expression appears as a 
 particular entry in the matrix product.
The matrix product can in turn be evaluated in logspace using results from~\cite{ChiuDL01}.

Balancing algorithms were first presented by Spira~\cite{Spira71} for Boolean formulas and by Brent~\cite{Brent74} for arithmetic expressions.
Later work showed that arithmetic expressions can be balanced in $\NC^1$ (observed implicitly in~\cite{BCGR92}) and in fact in $\TC^0$~\cite{GanardiL19}.
A generic framework for evaluating trees was presented in~\cite{KrebsLL17}, which implicitly balances the input tree in $\NC^1$.
The above-mentioned logspace version of Courcelle's theorem was improved to $\NC^1$~\cite{ElberfeldJT12} (under an appropriate input form) in subsequent work.
The first step of that algorithm is to balance a given tree decomposition in $\TC^0$.

In general, not every algebra admits such a depth-reduction result,
if one requires that the balanced tree is over the \emph{same} algebra~\cite[Theorem~1]{Kosaraju90}.
The core of most tree balancing approaches is a purely syntactic recursive decomposition of the input tree into subtrees
and contexts (subtrees where a subtree was removed) and the depth of this decomposition is bounded logarithmically
in the size of the input tree. Formally, this decomposition is
a tree straight-line program of logarithmic depth.
While subtrees evaluate to elements, contexts describe unary linear term functions over the algebra.
For example, over a commutative semiring a context computes an affine function $x \mapsto ax+b$,
and can be represented by the parameters $a,b$.
Furthermore, the composition of two affine functions can be implemented using semiring operations on these parameters.
The main challenge towards efficient tree balancing and tree evaluation over a particular algebra
is understanding the structure of its unary linear term functions (called the \emph{functional algebra} in~\cite{KrebsLL17}).
In general this can be difficult, as can be seen from the example of a finite algebra:
Given a tree automaton with $k$ states, the contexts can induce up to $k^k$ many state transformations.
In particular, the space bound of $\mathcal{O}(h \log \log k + \log k)$ achieved by the Cook-Mertz algorithm for an algebra of size $k$
cannot be immediately extended to unbalanced trees by applying the Cook-Mertz algorithm to a balanced tree straight-line program 
for the original tree, since $k$ would blow up to $k^k$.
For the special case of the Strahler algebra we provide a characterization of the unary linear term functions computed by contexts
in Section~\ref{sec-compute-st}.

A short version of this article appeared in \cite{GanardiL26}.
 
\section{Preliminaries} 

We assume some familiarity with formal language theory, in particular with context-free grammars; see e.g.~\cite{HoUl79} for details.
The set of all finite words over an alphabet $\Gamma$ is denoted with $\Gamma^*$; it includes the 
empty word $\varepsilon$. The length of a word $w \in \Gamma^*$ is $|w|$ and the number of occurrences
of $a \in \Gamma$ in the word $w$ is denoted with $|w|_a$.

\subsection{Directed acyclic graphs, trees, contexts}  \label{sec-dags-trees}

\subparagraph{Directed acyclic graph.}
We have to deal with node-labelled directed acyclic graphs (DAGs).
Let us fix a ranked alphabet $\Sigma$ (possibly infinite), where every $a \in \Sigma$ has a rank in $\mathbb{N}$.
Let $\Sigma_i \subseteq \Sigma$ be the set of symbols of rank $i \in \mathbb{N}$.
A  \emph{$\Sigma$-labelled DAG} is a tuple $\mathcal{D} = (V, v_0, \lambda, \gamma)$ with the following properties:
\begin{itemize}
\item $V$ is the finite set of nodes.
\item $v_0 \in V$ is a distinguished root.
\item $\lambda \colon V \to \Sigma$ is a mapping that assigns to every node $v \in V$ its label $\lambda(v)$.
We say that $v$ is a $\lambda(v)$-labelled node.
\item $\gamma \colon V \to V^*$ is a function such that $\lambda(v) \in \Sigma_{|\gamma(v)|}$.
It assigns to every node $v$ the list $\gamma(v)$ of $v$'s children (a node may occur more than once in this list).
\item We require that the directed graph $(V, \{ (u,v) : v \text{ appears in } \gamma(u) \})$ is acyclic.
\end{itemize} 
Sometimes we do not need the labelling function $\lambda$, in which case we omit $\lambda$ from the description of the DAG.

We also write $d(v) = |\gamma(v)|$ for the \emph{degree} of the node $v \in V$. Nodes of degree zero are also called \emph{leaves}.
For every $v \in V$ and $1 \leq i \leq d(v)$ we define $v \cdot i$ as the 
$i^{th}$ node in the word $\gamma(v)$.
This notation can be extended to words $\alpha \in \mathbb{N}^*$ (so-called \emph{address strings}) inductively:
$v \cdot \varepsilon = v$ and if $\alpha = \beta i$, $v \cdot \beta$ is defined and $1 \leq i \leq d(v \cdot \beta)$ then
$v \cdot \alpha = (v \cdot \beta) \cdot i$.
We define the \emph{size} $|\mathcal{D}|$ of $\mathcal{D}$ as $|\mathcal{D}| = \sum_{v \in V} (d(v)+1)$.

A path in $\mathcal{D}$ can be specified by its start node $v$
and an address string $\alpha = i_1 i_2 \cdots i_n \in \mathbb{N}^*$. The corresponding path consists
of the nodes 
$v, v\cdot i_1, v\cdot i_1i_2, \ldots, v\cdot \alpha$.
For a node $v \in V$ we define $\height_{\mathcal{D}}(v) = \max \{ |\alpha| : \alpha \in \mathbb{N}^*, v \cdot \alpha \text{ is defined} \}$.
Moreover, the \emph{height} (or \emph{depth}) of $\mathcal{D}$ is $\max \{ \height_{\mathcal{D}}(v) : v \in V\}$.

In the following, we will mainly consider \emph{binary} DAGs where $d(v) \leq 2$ for every $v \in V$.

\subparagraph{Trees.}
A  \emph{$\Sigma$-labelled tree} can be defined as a $\Sigma$-labelled DAG $t = (V, v_0, \lambda, \gamma)$
as above such that in addition for every $v \in V$ there is a unique address string $\alpha$ such that $v = v_0 \cdot \alpha$.
The node $v_0$ is the \emph{root} of the tree.
For a tree $t$ and a node 
$v$ we write $t(v)$ for the \emph{subtree} of $t$ rooted in $v$. It is the tree $(V', v, \lambda\rest_{V'}, \gamma\rest_{V'})$ where
$V' = \{ v \cdot \alpha : \alpha \in \mathbb{N}^*, v \cdot \alpha \text{ is defined}\}$.

From a DAG $\mathcal{D} = (V, v_0, \lambda, \gamma)$ one can define a tree $\unfold(\mathcal{D})= (V', \varepsilon, \lambda', \gamma')$ (the \emph{unfolding} of $\mathcal{D}$) as follows:
The set of nodes of $V'$ contains all address strings $\alpha$ such that $v_0\cdot\alpha$ is defined and the empty
string $\varepsilon$ is the root.
If $\alpha \in V'$ is such that $v = v_0 \cdot \alpha$,
then $\lambda'(\alpha) = \lambda(v)$ and $\gamma'(\alpha) = (\alpha 1) (\alpha 2) \cdots (\alpha d(v))$.
Figure~\ref{fig tree}(right) shows a DAG, whose unfolding is the tree on the left. The edge from a node to its $i^{th}$ 
child in the DAG is labelled with $i$. Clearly, the size of $\unfold(\mathcal{D})$ can be exponential in the size of $\mathcal{D}$.
This shows the potential of DAGs as a compact tree representation; see also \cite{DowneyST80,FlajoletSS90}.

Most trees in this paper are (unlabelled) \emph{binary trees}, in which case we have $d(v) \in \{0,2\}$ for all nodes $v$. 
For such trees we assume $\Sigma_0 = \{a\}$, and $\Sigma_2 = \{b\}$, i.e., internal nodes are labelled with $b$ and leaves are labelled with $a$.
Thus, the node labels do not carry any information and can be omitted.

\begin{figure}[t]
 \tikzset{alpha/.style={inner sep = 0pt, fill=white}}
  \tikzset{round/.style={inner sep = 1.2pt, circle, fill=black}}
\centering{
\begin{tikzpicture}
\node[alpha] (A) {$b$};
\node[alpha, below left = .8cm and 1.6cm of A] (A1) {$b$};
\node[alpha, below right = .8cm and 1.6cm  of A] (A2) {$b$};
\node[alpha, below left = .4cm and .8cm of A1] (A11) {$b$};
\node[alpha, below right = .4cm and .8cm  of A1] (A12) {$b$};
\node[alpha, below left = .4cm and .8cm of A2] (A21) {$b$};
\node[alpha, below right = .4cm and .8cm  of A2] (A22) {$b$};

\node[alpha, below left = .2cm and .4cm of A11] (A111) {$a$};
\node[alpha, below right = .2cm and .4cm  of A11] (A112) {$a$};
\node[alpha, below left = .2cm and .4cm of A12] (A121) {$b$};
\node[alpha, below right = .2cm and .4cm  of A12] (A122) {$a$};

\node[alpha, below left = .2cm and .4cm of A121] (A1211) {$a$};
\node[alpha, below right = .2cm and .4cm of A121] (A1212) {$a$};

\node[alpha, below left = .2cm and .4cm of A21] (A211) {$a$};
\node[alpha, below right = .2cm and .4cm  of A21] (A212) {$a$};
\node[alpha, below left = .2cm and .4cm of A22] (A221) {$b$};
\node[alpha, below right = .2cm and .4cm  of A22] (A222) {$a$};

\node[alpha, below left = .2cm and .4cm of A221] (A2211) {$a$};
\node[alpha, below right = .2cm and .4cm of A221] (A2212) {$a$};

\draw (A) -- (A1);
\draw (A) -- (A2);
\draw (A1) -- (A11);
\draw (A1) -- (A12);
\draw (A2) -- (A21);
\draw (A2) -- (A22);
\draw (A11) -- (A111);
\draw (A11) -- (A112);
\draw (A12) -- (A121);
\draw (A12) -- (A122);
\draw (A21) -- (A211);
\draw (A21) -- (A212);
\draw (A22) -- (A221);
\draw (A22) -- (A222);

\draw (A121) -- (A1211);
\draw (A121) -- (A1212);
\draw (A221) -- (A2211);
\draw (A221) -- (A2212);

\node[alpha,  right = 5cm of A] (B) {$b$};
\node[alpha, below = .6cm of B] (B1) {$b$};
\node[alpha, below = .6cm of B1] (B11) {$b$};
\node[alpha, below = .6cm of B11] (B111) {$b$};
\node[alpha, below = .6cm of B111] (B1111) {$a$};

\draw (B) to[bend left=30] node[right = -.65mm]{$\scriptstyle{2}$} (B1);
\draw (B) to[bend right=30] node[left = -.65mm]{$\scriptstyle{1}$} (B1);
\draw (B1) to[bend left=30] node[right = -.65mm]{$\scriptstyle{2}$}(B11);
\draw (B1) to[bend right=60] node[left = -.65mm]{$\scriptstyle{1}$} (B111);
\draw (B11) to[bend right=30] node[right = -.65mm]{$\scriptstyle{1}$} (B111);
\draw (B11) to[bend left=60] node[right = -.65mm]{$\scriptstyle{2}$} (B1111);
\draw (B111) to[bend left=30] node[right = -.8mm]{$\scriptstyle{2}$} (B1111);
\draw (B111) to[bend right=30] node[left = -.65mm]{$\scriptstyle{1}$} (B1111);

\end{tikzpicture}}
\caption{A binary tree.}  \label{fig tree}
\end{figure}

\subparagraph{Input representation of trees and DAGs.}
In circuit complexity (see the next section), the input representation of DAGs heavily influences the complexity. 
The representation as a tuple $(V, v_0, \lambda, \gamma)$ is also called \emph{pointer representation}.
In the pointer representation the edges of the DAG are given by adjacency list (namely the lists $\gamma(v)$).
In the case of trees, another well-known representation of a tree $t$ is the \emph{term representation}, where $t$ is represented by a term formed from the symbols in $\Sigma$.
For instance, the string $bbbaabbaaabbaabbaaa$ (which is written as
$b(b(b(a,a),b(b(a,a),a)),b(b(a,a),b(b(a,a),a)))$ for better readability) is the term representation
of the binary tree shown in Figure~\ref{fig tree}. It is obtained by listing the node labels of the binary tree in 
preorder.
With $\Bin$ we denote the set of all $x \in \{a,b\}^*$ that are the term representation of a binary tree.
It can be produced by the context-free grammar with the productions $S \to a \mid b S S$.

For binary DAGs, we also use the so-called extended connection representation, which extends the pointer representation
by a further relation; see also \cite{GanardiL19,Ruz81} and \cite[Definition~2.43]{Vollmer99}.
Consider a binary DAG $\mathcal{D} = (V, v_0, \lambda, \gamma)$ as above. 
The \emph{extended connection representation}, briefly \emph{ec-representation}, of $\mathcal{D}$, denoted by $\mathrm{ec}(\mathcal{D})$, is the tuple
$(V, v_0, \lambda,\gamma, \mathsf{ec}_{\mathcal{D}})$, where the set $\mathsf{ec}_{\mathcal{D}}$ consists of all  
so-called \emph{ec-triples} $(v, \alpha, v\cdot\alpha)$, where
$v \in V$, $\alpha \in \{1,2\}^*$ is an address string such that $v\cdot\alpha$ is defined and $|\alpha| \leq  \log_2 |\mathcal{D}|$.
Note that since $\mathcal{D}$ is binary, the number of   address strings with $|\alpha| \leq  \log_2 |\mathcal{D}|$ is bounded by 
$\mathcal{O}(|\mathcal{D}|)$.

\subparagraph{Contexts.}
Fix a so-called placeholder symbol $x \notin \{a,b\}$.
A \emph{binary context} is a binary tree $t$, where exactly one leaf $v \in V_0$ is labelled with
$x$. All other leaves are labelled with $a$ and internal nodes are labelled with $b$.
Given a binary context $t$ and a binary tree (resp., context) $t'$ we define 
the binary tree (resp., context) $t[x/t']$ by replacing the unique occurrence of $x$ in $t$ by $t'$.
For instance, we have $b(b(a,a),b(x,a))[x/b(a,a)] = b(b(a,a),b(b(a,a),a))$.

\subparagraph{Strahler numbers.} \label{page-st}
Let $s \colon \mathbb{N} \times \mathbb{N} \to \mathbb{N}$ be the binary operation with
\begin{equation} \label{eq-strahler-algebra}
s(x,y) = \begin{cases} 
x+1 & \text{ if } x = y, \\
\max(x,y) & \text{ if } x \neq y.
\end{cases}
\end{equation}
The algebraic structure $\mathcal{S} = (\mathbb{N}, s,0)$ is also called the \emph{Strahler algebra} in the following.
The Strahler number $\hs(t)$ of a binary tree $t \in \Bin$ is defined as follows:
\begin{eqnarray*}
\hs(a) & = & 0 \\
\hs(b(t_1, t_2)) & = & s(\hs(t_1), \hs(t_2)) 
\end{eqnarray*}
In other words:
$\hs(t)$ is obtained by evaluating $t$ in the Strahler algebra, where the binary symbol $b$ is interpreted by $s$ and the leaf symbol 
$a$ is interpreted by $0$. The tree from Figure~\ref{fig tree} has Strahler number 3.

It is well-known that if $t$ has $n$ leaves then $\hs(t) \leq \log_2 n$: Let $m = \hs(t)$. The case $m = 0$ is clear. 
If $m > 0$ then the root of $t$ must have at least two descendants 
with Strahler number $m-1$. By induction it follows that $t$ has at least $2^i$ many nodes with Strahler number $m-i$. Thus, $t$ has at least
$2^m$ many leaves (= nodes with Strahler number $0$). Moreover, 
the Strahler number of a tree $t$ is the largest $k$ such that a complete binary tree $t_k$ of depth $k$ can be embedded into $t$ (thereby,
edges of $t_k$ can be mapped to non-empty paths in $t$).

\subsection{Computational complexity} \label{sec-CC}

We assume that the reader is familiar with the complexity classes $\LS$ (deterministic logspace),
$\NL$ (nondeterministic logspace),
$\P$, $\NP$ and $\PSPACE$; see e.g.~\cite{AroBar09} for details.
A function $f \colon \Sigma^* \to \Gamma^*$ is  \emph{logspace computable} if it can be computed on a deterministic Turing-machine with
a read-only input tape, a write-only output tape and a working tape whose length is bounded logarithmically in the input length; such a machine is also called a logspace transducer. It is well known that the composition of logspace computable functions is logspace computable again.

In the rest of this section we briefly introduce some well-known concepts from circuit
complexity, more details can be found in the monograph \cite{Vollmer99}.

A \emph{(Boolean) circuit}  with $n$ inputs can be defined as a $\Sigma$-labelled DAG $\mathcal{B} = (V,v_0, \lambda, \gamma)$, where the set of node labels
$\Sigma$ consists of the symbols $x_1, \ldots, x_n$ of arity $0$ (the input variables) and additional 
Boolean functions of arbitrary arity (we identify here a $k$-ary Boolean function with a $k$-ary node label).
The set of these Boolean functions is also called the \emph{Boolean base} of $\mathcal{B}$.
Nodes of $\mathcal{B}$ are usually called \emph{gates} and the degree $d(v)$ of a gate $v$ is called its \emph{fan-in}.

A Boolean circuit $\mathcal{B} = (V,v_0, \lambda, \gamma)$ as above defines a mapping $\eta_\mathcal{B} : \{0,1\}^n \to \{0,1\}$ in the natural way:
Let $w = a_1 a_2 \cdots a_n \in \{0,1\}^n$. First define $\eta_v(w)$ for every gate $v$ inductively:
$\eta_v(w)=a_i$ if $\lambda(v) = x_i$ and $\eta_v(w) = f(\eta_{v_1}(w), \eta_{v_2}(w), \ldots, \eta_{v_d}(w))$ if 
$\gamma(v) = v_1 v_2 \cdots v_d$ and $\lambda(v) = f$ (a Boolean function of arity $d$).
Finally, we set $\eta_\mathcal{B}(w) = \eta_{v_0}(w)$.

The complexity class $\NC^1$ contains all languages $L \subseteq \{0,1\}^*$ such that there exists a circuit family $(\mathcal{B}_n)_{n \in \mathbb{N}}$ where
\begin{itemize}
\item $\mathcal{B}_n$ is a Boolean circuit with $n$ inputs over the Boolean base consisting of the unary function $\neg$ (negation) and 
the binary functions $\wedge$ (conjunction) and $\vee$ (disjunction),
\item $\mathcal{B}_n$ has size $n^{\mathcal{O}(1)}$ and depth $\mathcal{O}(\log n)$ and
\item for every $w \in \{0,1\}^n$, $\eta_{\mathcal{B}_n}(w)=1$ if and only if $w \in L$.
\end{itemize}
Important subclasses of $\NC^1$ are $\AC^0$ and $\TC^0$. The class $\AC^0$ is defined similarly to $\NC^1$ with the following modifications:
\begin{itemize}
\item The Boolean base of $\mathcal{B}_n$ consists of $\neg$ and disjunctions and conjunctions of any arity.
\item The depth of the circuit $\mathcal{B}_n$ is bounded by a fixed constant.
\end{itemize}
If one includes in the first point also majority functions of any
arity in the Boolean base, then one obtains the class $\TC^0$.
The $m$-ary majority function returns $1$ if and only if more than $m/2$ many input bits are $1$.

We only use the $\DLOGTIME$-uniform variants of $\AC^0$, $\TC^0$ and $\NC^1$.
For  $\AC^0$ and $\TC^0$, $\DLOGTIME$-uniformity means that
for a given tuple $(1^n, u, v)$, where $1^n$ is the unary encoding of $n \in \mathbb{N}$ and
 $u$ and $v$ are binary encoded gates of the $n$-th circuit $\mathcal{B}_n$, one can 
\begin{enumerate}[(i)]
\item compute the label of gate $u$ in time
$\mathcal{O}(\log n)$ and
\item \label{point-ii} check in  time $\mathcal{O}(\log n)$ whether $u$ is an input gate for $v$.
\end{enumerate}
Note that since the number of gates of $\mathcal{B}_n$ is polynomially bounded in $n$, the gates of $\mathcal{B}_n$ can be encoded by bit strings of length $\mathcal{O}(\log n)$. 
Thus the time bound $\mathcal{O}(\log n)$ is linear in the input length $|u|+|v|$.

The definition of $\DLOGTIME$-uniform $\NC^1$ is similar, but instead of \eqref{point-ii}
one requires that for given $1^n$, $u$, $v$ as above and an address string $\alpha \in \{1,2\}^*$ with
$|\alpha| \leq \log_2 |\mathcal{B}_n|$ one can check in time $\mathcal{O}(\log n)$ whether $u = v \cdot \alpha$
 \cite{BIS90,Ruz81,Vollmer99}. In other words, the relations from the ec-representation of $\mathcal{B}_n$ 
 can be verified in time $\mathcal{O}(\log n)$. We denote the $\DLOGTIME$-uniform variants of $\AC^0$, $\TC^0$ and $\NC^1$ with 
$\uAC^0$, $\uTC^0$ and $\uNC^1$, respectively.
  It is known that $\uNC^1$ coincides with $\mathsf{ALOGTIME}$
 (logarithmic time on an alternating random access Turing machine). The following inclusions hold between the complexity classes introduced above:
\[
\uAC^0 \subsetneq \uTC^0 \subseteq \uNC^1 = \mathsf{ALOGTIME} \subseteq \LS \subseteq \NL \subseteq \P \subseteq \NP \subseteq \PSPACE .
\]
The definitions of the above circuit complexity classes 
can be easily extended to functions $f\colon \{0,1\}^* \to \{0,1\}^*$.
This can be done by encoding $f$ by the language $L_f = \{ 1^i 0 w : w \in \{0,1\}^*, \text{ the $i$-th bit of $f(w)$ is 1} \}$.

Hardness for $\uNC^1$ (resp., $\uTC^0$) is always understood with respect to $\uTC^0$-computable (resp., $\uAC^0$-computable)
many-one reductions.

Typical problems in $\uTC^0$ are the computation of the integer quotient of binary encoded integers, and the sum and product of an arbitrary number of binary encoded integers \cite{HeAlBa02}. The canonical $\uTC^0$-complete language is
$\mathsf{Majority} = \{ x \in \{0,1\}^* : |x|_1 > |x|/2 \}$. Also the language $\Bin$ from Section~\ref{sec-dags-trees} is 
$\uTC^0$-complete. Membership in $\uTC^0$ was shown in \cite{Loh01rta}, and $\uTC^0$-hardness can be easily shown by a reduction
from the $\uTC^0$-complete language $\{ w\in \{0,1\}^* : |w|_0 = |w|_1 \}$.

A famous $\uNC^1$-complete problem is the Boolean formula problem: the input is a binary tree $t$ in term representation using the binary symbols
$\wedge$ and $\vee$ and the constant symbols $0$ (for false) and $1$ (for true), and the question is whether $t$ evaluates to $1$ in the 
Boolean algebra. Buss has shown the following theorem (note that the negation operator $\neg$ is not needed for $\uNC^1$-hardness 
in \cite{Bus87}):

\begin{theorem}[\cite{Bus87}] \label{thm buss}
The Boolean formula problem is complete for $\uNC^1$.
\end{theorem}
The following results are well-known and easy to show. Let $t$  be an arbitrary binary tree. 
\begin{itemize}
\item From the term representation of $t$ one can compute in $\uTC^0$ its pointer representation.
\item From the pointer representation of $t$ one can compute in logspace its term representation.
This transformation cannot be done in $\uNC^1$ unless $\LS = \uNC^1$ holds \cite{BeaudryM95}. 
\end{itemize}
The following lemma has been shown in \cite{GanardiL19}.

\begin{lemma} [\mbox{\cite[Lemma~3.4]{GanardiL19}}] \label{lemma-unfold}
For any $c > 0$ there exists a $\uTC^0$-computable function, which maps the ec-representation of a DAG $\mathcal{D}$ of size $n$
and depth at most $c \cdot \log_2 n$ to the term representation of the tree $\unfold(\mathcal{D})$.
\end{lemma}

\subsection{Tree straight-line programs} 
\label{sec-tslps}

In this section, we introduce \emph{tree straight-line programs (TSLPs)}, which have been studied mainly as a compressed representation of trees; see~\cite{Lohrey15} for a survey. Here, we define tree straight-line programs only for unlabelled binary trees.
A  tree straight-line program (TSLP) is a tuple $\mathcal{G} = (N_0, N_1, S, \rho)$ with the following properties:
\begin{itemize}
\item $N_0$ is a finite set of \emph{tree variables}. Tree variables are considered as symbols of rank~$0$.
\item $N_1$ is a finite set of \emph{context variables}. Context variables are considered as symbols of rank~$1$. Let $N = N_0 \cup N_1$.
We assume that $N_0 \cap N_1 = \emptyset$.
\item $S \in N_0$ is the \emph{start variable}.
\item $\rho$ maps every $A \in N_0$ to an expression $\rho(A)$ that has one of the following three forms, where
$B,C \in N_0$ and $D \in N_1$: $a$, $b(B,C)$, $D(C)$ (recall from Section~\ref{sec-dags-trees} that $a$ labels the leaves of a binary
tree and $b$ labels internal nodes).
\item $\rho$ maps every $A \in N_1$ to an expression $\rho(A)$ that has one of the following three forms, where
$B \in N_0$ and $C,D\in N_1$: $b(x,B)$, $b(B,x)$, $D(C(x))$ (here, $x$ is the placeholder symbol from contexts; see Section~\ref{sec-dags-trees}).
\item The binary relation $\{ (A,B) \in N \times N : B \text{ occurs in } \rho(A) \}$ must be acyclic.
\end{itemize}
For a TSLP $\mathcal{G} = (N_0, N_1, S, \rho)$ one should see the function $\rho$ as a set of term rewrite rules
$A \to \rho(A)$ for $A \in N$.
With these rewrite rules, we can derive from every $ A\in N_0$ (resp., $A \in N_1$) a binary tree
(resp., a binary context; see Section~\ref{sec-dags-trees}) $\val_{\mathcal{G}}(A)$ (the value of $A$). We omit the index $\mathcal{G}$ if it is clear
from the context. Formally, we define $\val_{\mathcal{G}}(A)$ as follows:
\begin{itemize}
\item if $A \in N_0$ and $\rho(A)=a$ then $\val_{\mathcal{G}}(A) = a$,
\item if $A \in N_0$ and $\rho(A)=b(B,C)$ then $\val_{\mathcal{G}}(A) = b(\val_{\mathcal{G}}(B),\val_{\mathcal{G}}(C))$,
\item if $A \in N_0$ and $\rho(A)=D(C)$ then $\val_{\mathcal{G}}(A) = \val_{\mathcal{G}}(D)[x/\val_{\mathcal{G}}(C)]$,
\item if $A \in N_1$ and $\rho(A)=b(x,B)$ then $\val_{\mathcal{G}}(A) = b(x,\val_{\mathcal{G}}(B))$,
\item if $A \in N_1$ and $\rho(A)=b(B,x)$ then $\val_{\mathcal{G}}(A) = b(\val_{\mathcal{G}}(B),x)$,
\item if $A \in N_1$ and $\rho(A)=D(C(x))$ then $\val_{\mathcal{G}}(A) = \val_{\mathcal{G}}(D)[x/\val_{\mathcal{G}}(C)]$.  
\end{itemize}
Finally, we define the binary tree $\val(\mathcal{G}) = \val_{\mathcal{G}}(S)$ (recall that $S \in N_0$).

\begin{example}
Consider the TSLP $\mathcal{G}$  with $N_0 = \{S, A, B, C, D\}$, $N_1 = \{ E \}$ and the following rules:
$S \to b(A,A), \ A \to b(B,C), \ C \to E(B), \ B \to E(D), \ E(x) \to b(x,D), \ D \to a$.
Then $\val(\mathcal{G})$ is the tree from Figure~\ref{fig tree}.
\end{example}
A TSLP $\mathcal{G} = (N_0, N_1, S, \rho)$ can be encoded by a $\Sigma$-labelled DAG $(N_0 \cup N_1, S, \lambda, \gamma)$
with $\Sigma_0 = \{ a\}$, $\Sigma_1 = \{b_1, b_2 \}$ and $\Sigma_2 = \{ b, \circ_0, \circ_1 \}$ 
 in the following way:
\begin{itemize}
\item if $A \in N_0$ and $\rho(A)=a$ then $\lambda(A) = a$ and $\gamma(A) = \varepsilon$,
\item if $A \in N_0$ and $\rho(A)=b(B,C)$ then $\lambda(A) = b$ and $\gamma(A) = BC$,
\item if $A \in N_0$ and $\rho(A)=D(C)$ then $\lambda(A) = \circ_0$ and $\gamma(A) = DC$, 
\item if $A \in N_1$ and $\rho(A)=b(x,B)$ then $\lambda(A) = b_1$ and $\gamma(A) = B$,
\item if $A \in N_1$ and $\rho(A)=b(B,x)$ then $\lambda(A) = b_2$ and $\gamma(A) = B$,
\item if $A \in N_1$ and $\rho(A)=D(C(x))$ then $\lambda(A) = \circ_1$ and $\gamma(A) = DC$.   
\end{itemize}
In particular, we can speak about the ec-representation of a TSLP or the height of a variable in a TSLP.
We define the size $|\mathcal{G}|$ of the TSLP $\mathcal{G}$ as the size of the corresponding DAG, which is bounded
by $3 |N|$. It is easy to see that the tree $\val(\mathcal{G})$ has at most $2^{\mathcal{O}(|\mathcal{G}|)}$ many nodes.

Note that for a TSLP $\mathcal{G}$, where $N_1 = \emptyset$ (hence, every $\rho(A)$ is either $a$ or $b(B,C)$ for $B,C \in N_0$),
the unfolding of the above DAG is $\val(\mathcal{G})$. In general, TSLPs can be more succinct than DAGs: take for instance a caterpillar tree
$t = b(b(\ldots b(a,a),a), \ldots, a)$ of size $n$. It can be represented by a TSLP of size $\mathcal{O}(\log n)$, whereas every DAG that unfolds into $t$ has size $\Omega(n)$.
The following result from  \cite{GanardiL19} will be important in the next section.

\begin{theorem}[\mbox{\cite[Theorem~5.6]{GanardiL19}}] \label{theorem tree->tslp}
From a binary tree $t$ of size $n$ given in term representation one can compute in $\uTC^0$ 
the ec-representation of a TSLP 
$\mathcal{G}$ of depth $\mathcal{O}(\log n)$ and size $\mathcal{O}(n)$ such that $\val(\mathcal{G}) = t$.
\end{theorem}
The size bound $\mathcal{O}(n)$ for the TSLP $\mathcal{G}$ in Theorem~\ref{theorem tree->tslp} can be even replaced by $\mathcal{O}(n/\log n)$ \cite[Theorem~5.6]{GanardiL19}, but this is not important for our purpose.

\section{Complexity of computing the Strahler number} \label{sec-compute-st}

In this section we consider the problem of checking whether the Strahler number of a given binary tree is at least a given threshold.
The problem $\StN$ is defined as follows:
\begin{itemize}
\item Input: a binary tree $t$ and a number $k$.
\item Question: Is $\hs(t) \geq k$?
\end{itemize}
If we fix the value $k \geq 1$, then we obtain the following problem $\StNk$:
\begin{itemize}
\item Input: a binary tree $t$.
\item Question: Is $\hs(t) \geq k$?
\end{itemize}
These problem descriptions are actually incomplete, since we did not fix the input encoding of $t$, which influences
the complexity of the problems.
We obtain the following variations: In $\StNterm$ (resp., $\StNpointer$) the tree $t$ is given 
by its term (resp., pointer) representation. In $\StNdag$ (resp., $\StNtslp$) the tree $t$
is given succinctly by a binary DAG $\mathcal{D}$ (resp., a TSLP $\mathcal{G}$) such that $t = \unfold(\mathcal{D})$ (resp., $t = \val(\mathcal{G})$).
The problems $\StNkterm$, $\StNkpointer$, $\StNkdag$, and $\StNktslp$ are defined analogously.
Our main result is:

\begin{theorem} \label{thm-main-NC1}
$\StNterm$ is $\uNC^1$-complete.
\end{theorem}
As a gentle introduction to the problem, we first present a weaker result, namely that one can
calculate the Strahler number of a tree with $n$ leaves in $\mathcal{O}(\log n \log \log n)$ space,
and then show how to reduce the space complexity to $\mathcal{O}(\log n)$.
We only sketch the proof, since these space bounds are subsumed by the $\uNC^1$ upper bound,  proven later in this section.

The idea is to compute the Strahler number recursively by traversing the tree, in a depth-first order.
We perform a depth-first traversal through the tree, maintaining a single pointer to the current node using $\mathcal{O}(\log n)$ space,
and a constant-sized information, indicating the direction of the next traversal step.
Additionally, we store a list of the Strahler numbers $s_1, \dots, s_k$ of the maximal subtrees by inclusion $t_1, \dots, t_k$
that have been completely traversed in that order (i.e., $s_i$ has been traversed before $s_j$ for $i < j$).
Whenever both subtrees of a node have been traversed, we can combine their Strahler numbers to obtain the Strahler number
of the parent node.
Since each Strahler number $s_i$ is bounded by $\log n$ where $n$ is the number of leaves in the input tree, it can be stored
in  $\mathcal{O}(\log \log n)$ bits. However, if the tree is traversed in an arbitrary order, the number $k$ of subtrees could be up to linear in $n$.
The solution is to visit \emph{heavy} subtrees first, i.e.\ if a node is visited for the first time, the next step moves to the larger subtree of the current node (if both subtrees have the same size, one moves to the left subtree).
Note that the size of a subtree can be computed in logspace.
This ensures that $|t_i| \ge |t_{i+1}| + \dots + |t_k|$ and therefore $|t_i| \ge 2|t_{i+2}|$ for $i \leq k-2$. In particular, $k$ is bounded by $\mathcal{O}(\log n)$
and the total space complexity is $\mathcal{O}(k \log \log n)= \mathcal{O}(\log n \log \log n)$.

To shave off the $\log \log n$ factor, we store the sequence of Strahler numbers $s_1, \dots, s_k$ using a delta encoding.
Let us say that a number $s_i$ is \emph{dominated} if there exists $j>i$ such that $s_j > s_i$.
Such a dominated number $s_i$ can be replaced by $0$, which does not change the Strahler number of the tree.
The subsequence of undominated numbers $s_{i_1}, s_{i_2}, \dots, s_{i_\ell}$ is monotonically decreasing
and can be represented using a delta encoding $s_{i_1}-s_{i_2},\ldots, s_{i_{\ell-1}}-s_{i_\ell}, s_{i_\ell}$. Each number
$s$ in this sequence will be represented in unary encoding by $1^s \#$. As an example, the sequence $s_1, \ldots, s_k \ = \ 3,2,5,3,4,4,2,1$ is encoded
by  $001\#0\#11\#1\#1\#$.
The resulting word over the alphabet $\{0,1, \# \}$ has length $k + s_{i_1} \le \mathcal{O}(\log n)$.

Now we turn to proving the $\uNC^1$ bound from Theorem~\ref{thm-main-NC1}. To this end, we will first compute from the input tree $t$ a TSLP using
Theorem~\ref{theorem tree->tslp}. To make use of the TSLP-representation of $t$, we need a simple description of 
unary linear term functions in the Strahler algebra. For this, we start with some preparations.
 
Consider a TSLP $\mathcal{G} = (N_0, N_1, S, \rho)$ as defined in Section~\ref{sec-tslps}. For a tree variable $A \in N_0$ we write 
$\hs_A$ for the Strahler number $\hs(\val(A))$.
For a context variable $B \in N_1$ we define a function $\hs_B : \mathbb{N} \to \mathbb{N}$  as
follows: Consider the binary context $t = \val(B)$. Take an integer $n \in \mathbb{N}$ and take a binary
tree $t'$ with $\hs(t') = n$. The concrete choice of $t'$ is not important. Then we define $\hs_B(n) = \hs(t[x/t'])$.
Intuitively speaking, the evaluation of the binary context $t$ in the Strahler algebra yields the unary linear term function $\hs_B$.
If we substitute the placeholder $x$ by a number $n$ then we can evaluate the resulting expression in the Strahler
algebra and the result is $\hs_B(n)$. The functions $\hs_B$ can be described by two integers in the following way:
For $\ell,h \in \mathbb{N}$ with $0 \le \ell \leq h$ we define the function $[\ell,h] :  \mathbb{N} \to \mathbb{N}$ as follows:
\begin{equation} \label{eq-l-h}
[\ell,h](x) = \begin{cases}
h & \text{ if } x < \ell \\
h+1 & \text{ if }  \ell \le x \le h \\
x & \text{ if }  x >  h
\end{cases}
\end{equation}

\begin{lemma} \label{lemma-composition}
The set of functions $[\ell,h]$ is closed under composition. More precisely, for all $\ell \leq h$, $m \le i$ and $x \in \mathbb{N}$ we have the following:
\begin{equation} \label{eq-comp}
[m,i]([\ell,h](x)) = \begin{cases}
[m,i](x) & \text{ if } h+2 \leq m \\
[\ell,i](x) & \text{ if } h+1 = m  \\ 
[0,i](x) & \text{ if } m \le h \leq i \\
[\ell,h](x) & \text{ if } i < h 
\end{cases}
\end{equation}
\end{lemma}

\begin{proof}
For all $x \in \mathbb{N}$ we have:
\begin{equation} \label{eq-comp2}
[m,i]([\ell,h](x)) = \begin{cases}
[m,i](h) & \text{ if } x < \ell \\
[m,i](h+1) & \text{ if } \ell \le x \le h  \\ 
[m,i](x) & \text{ if } x > h  
\end{cases}
\end{equation}
We now distinguish the four cases from \eqref{eq-comp}:

\medskip
\noindent
\emph{Case 1.} $h+2 \leq m$: We have to show that $[m,i]([\ell,h](x)) = [m,i](x)$ for all $x$.
Since $h+1 < m$ we obtain from \eqref{eq-comp2}:
\begin{equation*} 
[m,i]([\ell,h](x)) = \begin{cases}
i = [m,i](x) & \text{ if } x \le h \\
[m,i](x) & \text{ if } x > h  
\end{cases}
\end{equation*}
\emph{Case 2.} $h+1 = m$: We have to show $[m,i]([\ell,h](x)) = [\ell,i](x)$ for all $x$.
With  \eqref{eq-comp2} we get
\begin{equation*} 
\left.
[m,i]([\ell,h](x)) \ = \ \begin{cases}
[h+1,i](h) = i & \text{ if } x < \ell \\
[h+1,i](h+1) = i+1& \text{ if } \ell \le x \le h  \\ 
[h+1,i](x) = i+1 & \text{ if } h < x \leq i \\
[h+1,i](x) = x & \text{ if } x > i
\end{cases}
 \right\} = [\ell,i](x).
\end{equation*}
\emph{Case 3.} $m \le h \leq i$. We have to show $[m,i]([\ell,h](x)) = [0,i](x)$ for all $x$.
Equation  \eqref{eq-comp2} yields
\begin{equation*} 
[m,i]([\ell,h](x)) = \begin{cases}
i+1 = [0,i](x)& \text{ if } x \le i, \\
x = [0,i](x) & \text{ if } x > i.
\end{cases}
\end{equation*}
\emph{Case 4.} $i < h$:  We have to show that $[m,i]([\ell,h](x)) = [\ell,h](x)$ for all $x$.
Equation \eqref{eq-comp2} simplifies for $i < h$ to
\begin{equation*} 
\left.
[m,i]([\ell,h](x)) = \begin{cases}
h & \text{ if } x < \ell \\
h+1 & \text{ if } \ell \le x \le h \\ 
x & \text{ if } x > h
\end{cases}
 \right\} = [\ell,h](x) .
\end{equation*}
This concludes the proof of the lemma.
\end{proof}

\begin{lemma} \label{lemma-l_B-h_B}
Let $\mathcal{G} = (N_0, N_1, S, \rho)$ be a TSLP and let $A \in N_1$. Then there exist numbers $\ell_A$ and $h_A$ such that 
$\hs_A = [\ell_A, h_A]$.
\end{lemma}

\begin{proof}
Every context $\val(A)$ for $A \in N_1$ can be obtained from composing contexts of the form $b(x,t)$ and $b(t,x)$
where $t = \val(B)$ for some $B \in N_0$. By Lemma~\ref{lemma-composition} it therefore suffices to show that for every $m \in \mathbb{N}$ the
mapping $x \mapsto s(x,m)$ (where $s$ is from \eqref{eq-strahler-algebra})  is of the form $[\ell,h]$ (then, since $s$ is commutative,
the same holds for the mapping  $x \mapsto s(m,x)$). It is straightforward to check that
$s(x,m) = [m,m](x)$ 
for all $x \in \mathbb{N}$.
\end{proof}
We are now in the position to prove Theorem~\ref{thm-main-NC1}.

\begin{proof}[Proof of Theorem~\ref{thm-main-NC1}]
We start with the $\uNC^1$ upper bound.
Let $t \in \Bin$ be a binary tree given in term representation. Let $n$ be the number of leaves of $t$. Hence, we must have $\hs(t) \leq \log_2 n$.
Our goal is to compute in $\uTC^0$ from $t$ and an integer $0 \le k \le \log_2 n$ a Boolean circuit $\mathcal{B}_{t,k}$ of depth $\mathcal{O}(\log n)$ such that
$\mathcal{B}_{t,k}$ evaluates to true if and only if $\hs(t) \ge k$. 
The Boolean circuit $\mathcal{B}_{t,k}$ is represented in ec-representation, which ensures that it can be unfolded
in $\uTC^0$ into an equivalent Boolean formula (Lemma~\ref{lemma-unfold}) and then evaluated in $\uNC^1$ by Theorem~\ref{thm buss}.

In a first step, we use Theorem~\ref{theorem tree->tslp} to
compute in $\uTC^0$ the ec-representation of a TSLP 
$\mathcal{G} = (N_0, N_1, S, \rho)$ of depth $\mathcal{O}(\log n)$ and size $\mathcal{O}(n)$ such that $\val(\mathcal{G}) = t$.
Let $N = N_0 \cup N_1$. Since the ec-representation of $\mathcal{G}$ is available and 
the depth of $\mathcal{G}$ is bounded by $\mathcal{O}(\log n)$, one can ensure in $\uTC^0$ that
all variables in $N$ can be reached from the start variable $S$ (this property is actually satisfied when $\mathcal{G}$ is constructed
according to \cite{GanardiL19}). In particular, every variable $A \in N_0$ produces a subtree of $t$ and every variable $A \in N_1$
produces a subcontext of $t$.
Hence, every number $\hs_A$ is bounded by $\log_2 n$ and we will see in a moment that the same holds for the numbers $\ell_A$ and $h_A$
from Lemma~\ref{lemma-l_B-h_B}.

In the following we consider the following set of formal integer variables:
$$
\Delta(\mathcal{G}) = \{ A_\hs : A \in N_0 \} \cup \{ A_\ell , A_h : A \in N_1 \} .
$$
For $X \in \Delta(\mathcal{G})$ we define the integer $v(X)$ by $v(A_\hs) = \hs_A$, 
$v(A_\ell) = \ell_A$ and $v(A_h) = h_A$.

We will define a Boolean circuit $\mathcal{B}_{t,k}$ that contains for all $i \in \mathbb{Z}$ with 
$|i| \leq \log_2 n$ and 
all $X, Y \in \Delta(\mathcal{G})$ a gate $[X \leq Y + i]$ with the obvious meaning: the gate evaluates to true if and only if
$v(X) \leq v(Y) + i$. 
Let $V$ be the set of all such gates $[X \leq Y + i]$.
It is convenient to allow also gates $[X \leq i]$ and $[X \geq i]$. Formally, they can be replaced
by $[X \leq A_\hs + i]$ and $[A_\hs \leq X - i]$, where $A \in N_0$ is a variable with $\rho(A) = a$ (so that $\hs_A = 0$).
The output gate of $\mathcal{B}_{t,k}$ is $[S_\hs \ge k]$.

The number $i$ is called the \emph{offset} of the gate $g = [X \leq Y+i]$ 
and we define $N(g) = \{A,B\}$ if  $X \in \{ A_\hs, A_\ell, A_h\}$ and $Y \in \{ B_\hs, B_\ell, B_h\}$. For gates $g_1, g_2 \in V$, we write
$g_1 \succ g_2$ if every variable $B \in N(g_2)$ appears in $\rho(A)$ for some $A \in N(g_1)$. Then the length of every chain 
$g_1 \succ g_2 \succ g_3 \succ \cdots \succ  g_m$ is bounded by the depth of $\mathcal{G}$, which is $\mathcal{O}(\log n)$.

To define the wires of $\mathcal{B}_{t,k}$, 
first note that the numbers $v(X)$ ($X \in \Delta(\mathcal{G})$)
are computed according to the following rules:
\begin{enumerate}[(i)]
\item \label{evaluate-1} if $A \in N_0$ and $\rho(A)=a$ then $\hs_A = 0$,
\item \label{evaluate-2} if $A \in N_0$ and $\rho(A) = b(B,C)$ then 
\[
\hs_A = \begin{cases}
\hs_B & \text{ if } \hs_B > \hs_C, \\
\hs_C &\text{ if } \hs_B < \hs_C, \\
\hs_B+1 & \text{ if } \hs_B = \hs_C,
\end{cases}
\]
\item \label{evaluate-3} if $A \in N_0$ and $\rho(A) = B(C)$ then, since $\hs_B = [\ell_B, h_B]$,
\[
\hs_A = \begin{cases}
h_B & \text{ if } \hs_C < \ell_B, \\
h_B+1 &\text{ if } \ell_B \le \hs_C \le h_B, \\
\hs_C & \text{ if } \hs_C > h_B,
\end{cases}
\]
\item \label{evaluate-4} if $A \in N_1$ and $\rho(A) = b(x,B)$ or $\rho(A) = b(B,x)$ then
$\ell_A = h_A = \hs_B$, and
\item \label{evaluate-5} if $A \in N_1$ and $\rho(A) = B(C(x))$ then by Lemma~\ref{lemma-composition} we have:
\begin{equation} \label{case-B(C(x))}
\ell_A = \begin{cases}
\ell_B & \text{ if } h_C + 2 \le \ell_B, \\
\ell_C & \text{ if }  h_C + 1 = \ell_B, \\
0 & \text{ if $\ell_B \le h_C \le h_B$,} \\
\ell_C & \text{ if } h_B < h_C,
\end{cases} \qquad
h_A = \begin{cases}
h_B & \text{ if } h_C + 2 \le \ell_B, \\
h_B & \text{ if }  h_C + 1 = \ell_B, \\
h_B & \text{ if } \ell_B \le h_C \le h_B, \\
h_C & \text{ if } h_B < h_C.
\end{cases}
\end{equation}
\end{enumerate}
Points \eqref{evaluate-4} and \eqref{evaluate-5} imply that all numbers $\ell_A$ and $h_A$ for $A \in N_1$ are equal to some $\hs_B$ ($B \in N_0$) and therefore bounded by $\log_2 n$.

From the equalities in \eqref{evaluate-1}--\eqref{evaluate-5}, it is now straightforward to construct for every gate $g \in V$ 
 a Boolean circuit $\mathcal{B}_g$ of constant size with output gate $g$. All input gates $g'$ of 
  $\mathcal{B}_g$ satisfy $g \succ g'$.
Let us consider for instance the gate $[A_h \leq D_\ell + i]$ for $A, D \in N_1$ and assume that $\rho(A) = B(C(x))$ and 
$\rho(D) = b(x,E)$. Then, $\ell_D = \hs_E$ and the equation for $h_A$ in \eqref{case-B(C(x))} implies
\[
h_A \leq \ell_D + i \  \Longleftrightarrow  \ (h_B \leq \hs_E + i \wedge h_C  \le h_B ) \vee  (h_C \leq \hs_E + i \wedge h_B < h_C).
\]
This equivalence directly yields the Boolean circuit for $[A_h \leq D_\ell + i]$. Its input gates are: 
$[B_h \leq E_\hs + i]$, $[C_h  \le B_h]$,  $[C_h \leq E_\hs + i]$, and $[B_h \le C_h-1]$. 

When constructing a Boolean circuit $\mathcal{B}_g$ one may obtain gates, where the absolute value of the offset $i$ is larger than $\log_2 n$.
Such gates can be replaced by \textsf{true} or \textsf{false}: $[X \leq Y + i]$ with $i > \log_2 n$ 
can be replaced by \textsf{true} (since $v(X) \leq \log_2 n $ and $v(Y) \geq 0$)
and $[X \leq Y - i]$ with $i > \log_2 n$  can be replaced by \textsf{false}.

The circuit $\mathcal{B}_{t,k}$ results from the union of the above constant-size circuits. 
Since the depth of $\mathcal{G}$ is bounded by $\mathcal{O}(\log n)$, it follows that the depth
of $\mathcal{B}$ is also bounded by $\mathcal{O}(\log n)$. Moreover, the ec-representation of $\mathcal{B}$ can 
be easily computed in $\uTC^0$ from the ec-representation of the TSLP $\mathcal{G}$. 
This shows the upper bound from Theorem~\ref{thm-main-NC1}.

For the lower bound we give a reduction from the $\uNC^1$-complete 
Boolean formula value problem; see Theorem~\ref{thm buss}. Binary conjunction $\wedge$ is simulated by the operation 
\begin{equation} \label{eq-conjunction}
f_\wedge(x,y) = s(x+1,y+1) = s(s(x,x),s(y,y)),
\end{equation}
where $s$ from \eqref{eq-strahler-algebra}, and binary disjunction $\vee$ is simulated by the operation
\begin{equation} \label{eq-disjunction}
f_\vee(x,y) = s(s(x+1,y),s(x,y+1)) = s(s(s(x,x),y),s(x,s(y,y))).
\end{equation}
We obtain for every $a \geq 0$:
\begin{align}
 & f_\wedge(a,a) = f_\wedge(a,a+1) = f_\wedge(a+1,a) = a+2 \text{ and } f_\wedge(a+1,a+1) = a+3, \label{eq-wedge} \\
 & f_\vee(a,a) = a+2 \text{ and } f_\wedge(a,a+1) = f_\wedge(a+1,a) =  f_\wedge(a+1,a+1) = a+3. \label{eq-vee}
\end{align} 
A given Boolean formula (built from binary operators $\wedge$ and $\vee$; negation is not needed in \cite{Bus87})
can be transformed in $\uTC^0$ into an equivalent Boolean formula $\Phi$ of depth
$d \leq \mathcal{O}(\log |\Phi|)$; see \cite{GanardiL19}.
We can also assume that every path from the root to a leaf has the same length $d$.
By replacing in $\Phi$ every $\wedge$ (resp., $\vee$)
by $f_\wedge$ (resp., $f_\vee$) and replacing 
every occurrence of the truth value $\mathsf{true}$ (resp., $\mathsf{false}$) by $1= s(0,0)$ (resp., $0$),
we obtain an expression that evaluates in the Strahler algebra to $2d+1$ (resp., $2d$) if the Boolean formula $\Phi$ evaluates to 
$\mathsf{true}$ (resp., $\mathsf{false}$). Note that replacing  $f_\wedge(x,y)$ and $f_\vee(x,y)$ by their left-hand sides
from \eqref{eq-conjunction} and \eqref{eq-disjunction}
yields a DAG (since $x$ and $y$ appear more than once in \eqref{eq-conjunction} and \eqref{eq-disjunction}) whose ec-representation can be computed in $\uTC^0$ from $\Phi$.
Since the depth of this DAG is $\mathcal{O}(\log |\Phi|)$ it can be unfolded into the 
term representation of a tree in $\uTC^0$ by Lemma~\ref{lemma-unfold}.
\end{proof}
For $\StNkterm$ with $k \geq 4$ we can show $\uTC^0$-completeness via a reduction from $\mathsf{Majority}$:
\begin{theorem} \label{theorem-uTC0}
The problem $\StNkterm$  is $\uTC^0$-complete  for every $k \geq 4$.
In particular, there is a $\uAC^0$-computable function $t : \{0,1\}^* \to \Bin$ such that the following holds for every $w \in \{0,1\}^*$:
if $w \in \mathsf{Majority}$ then $\hs(t(w)) = 4$, otherwise  $\hs(t(w)) = 3$.
\end{theorem}

\begin{proof}[Proof of Theorem~\ref{theorem-uTC0}]
We first show that every problem $\mathsf{St}_{\mathsf{term}}^{\ge k}$ (for a fixed $k$) is in $\uTC^0$. 
Let $t \in \Bin$. We have to check whether the complete binary tree $t_k$ of depth $k$ embeds into $t$.
For this, we have to check whether there exist $2^{k+1}-1$ different positions in $t$ such that the corresponding
nodes are in the correct descendant relations in order to yield an embedding of $t_k$. Hence, it suffices
to show that in  $\uTC^0$ one can check whether for two positions $i < j$ in $t$ (that are identified with 
 the corresponding tree nodes), $j$ is a proper descendant of $i$. This holds if and only if there exists a position $k \geq j$ in $t$ such
that $t[i,k]$ (the substring of $t$ starting in position $i$ and ending in position $k$) belongs to $\Bin$. Since
$\Bin$ belongs to $\uTC^0$ \cite{Loh01rta}, this can be checked in $\uTC^0$ as well.

For the hardness part it suffices to consider the case $k=4$.
We start with the morphism $d :  \{0,1\}^* \to \{0,1\}^*$ that doubles each symbol: $d(0) = 00$ and $d(1) = 11$.
Clearly, $w \in \mathsf{Majority}$ if and only if $d(w) \in \mathsf{Majority}$ and $d(w)$ can be computed in $\uAC^0$.
Hence it suffices to consider strings in the image of $d$ ($\mathsf{im}(d)$ for short) in the following. Every $w \in \mathsf{im}(d)$ 
can be uniquely factorized as 
$w = 0^{2k_0} 1 0^{2k_1} 1  0^{2k_2} \cdots 1 0^{2k_m}$, where $m = |w|_1 \geq 0$ (which is even) 
and $k_1, \ldots, k_m \geq 0$ (actually, every second $k_i$ is zero due to the factors $11$ but we do not need this fact in the following).
Then we define
$$
f(w) \ = \ b^{k_0} a^{k_0} b^{1+k_1} a^{k_1} b^{1+k_2} a^{k_2} \cdots b^{1+k_m} a^{k_m}.
$$
In other words: every $1$ in $w$ is replaced by $b$ and every maximal block $0^{2 k}$ of zeros in $w$ is replaced by $b^k a^k$.
\begin{claim}
The function $f$ can be computed in $\uAC^0$.
\end{claim}

 \begin{claimproof}
To see this, note that  $f(w)$ is obtained from $w$ by replacing every 
 $1$ by $b$ and every $0$ by either $b$ or $a$ according to the following rule:
 Assume that an occurrence of $0$ is the $i^{\text{th}}$ $0$ in a maximal block of $2k$ zeros. Then this occurrence of $0$ is replaced by
 $b$ if $i \leq k$ and otherwise by $a$.
 This case distinction can be easily implemented by a bounded depth Boolean circuit of unbounded fan-in.
  \end{claimproof}
Finally, for a bit string $w \in \mathsf{im}(d)$ of length $2n$ we define
$$
t(w) \ = \ f(w) \, b  \, T_2  \, f(\bar{w})  \, T_2  \, a^{n-1}  \, T_3  \, a^n ,
$$
where $\bar{w} \in \mathsf{im}(d)$ is obtained by flipping every bit in $w$,
$T_2 = bbaabaa$ (the term representation of a tree of Strahler number 2) and $T_3 = bT_2 T_2$ (the term representation of a tree of Strahler number $3$).
Since $f$ can be computed in $\uAC^0$, the same holds for $t$.

The string $t(w)$ is the term representation of a binary tree and satisfies the following:

\begin{claim} \label{claim-tc0}
If $|w|_0 \ge n$ then $\hs(t(w)) = 3$, otherwise  $\hs(t(w)) = 4$.
\end{claim}
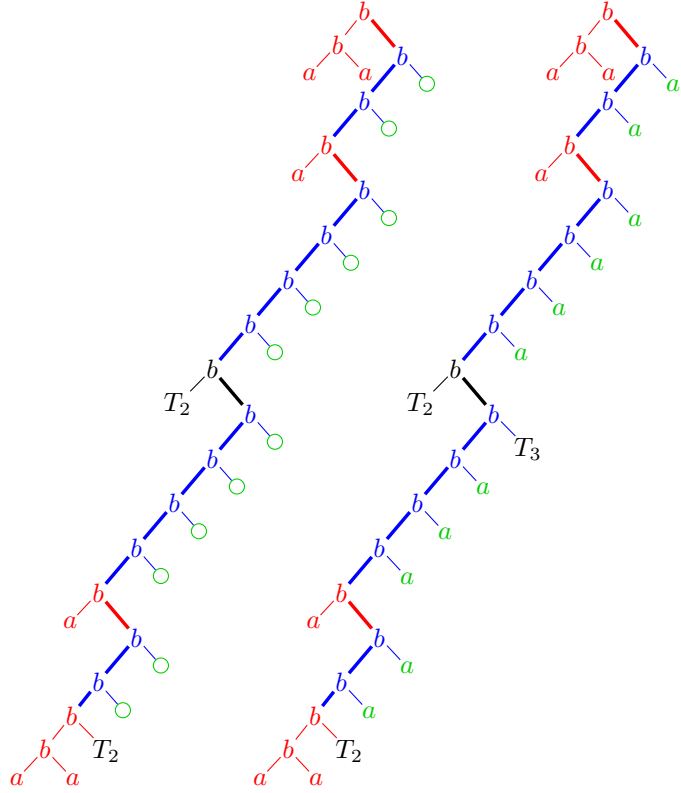
\begin{figure}[t]
 \tikzset{alpha/.style={inner sep = 0.5pt, fill=white}}
  \tikzset{round/.style={inner sep = 2pt, draw = darkgreen, circle, fill=white}}
\centering{
\begin{tikzpicture}
\node[alpha] (1) {$\textcolor{red}{b}$};
\node[alpha, below left = .15cm and .15cm of 1] (1') {$\textcolor{red}{b}$};
\node[alpha, below left = .15cm and .15cm  of 1'] (1a) {$\textcolor{red}{a}$};
\node[alpha, below right = .15cm and .15cm  of 1'] (1aa) {$\textcolor{red}{a}$};
\node[alpha, below right = .30cm and .30cm of 1] (2) {$\textcolor{blue}{b}$};
\node[round, below right = .15cm and .15cm of 2] (h2) {};
\node[alpha, below left = .30cm and .30cm of 2] (3) {$\textcolor{blue}{b}$};
\node[round, below right = .15cm and .15cm of 3] (h3) {};
\node[alpha, below left = .30cm and .30cm of 3] (4) {$\textcolor{red}{b}$};
\node[alpha, below left = .15cm and .15cm  of 4] (4a) {$\textcolor{red}{a}$};
\node[alpha, below right = .30cm and .30cm of 4] (5) {$\textcolor{blue}{b}$};
\node[round, below right = .15cm and .15cm of 5] (h5) {};
\node[alpha, below left = .30cm and .30cm of 5] (6) {$\textcolor{blue}{b}$};
\node[round, below right = .15cm and .15cm of 6] (h6) {};
\node[alpha, below left = .30cm and .30cm of 6] (7) {$\textcolor{blue}{b}$};
\node[round, below right = .15cm and .15cm of 7] (h7) {};
\node[alpha, below left = .30cm and .30cm of 7] (8) {$\textcolor{blue}{b}$};
\node[round, below right = .15cm and .15cm of 8] (h8) {};
\node[alpha, below left = .30cm and .30cm of 8] (9) {$b$};
\node[alpha, below right = .30cm and .30cm of 9] (10) {$\textcolor{blue}{b}$};
\node[alpha, below left = .15cm and .15cm of 9] (T2) {$T_2$};
\node[round, below right = .15cm and .15cm of 10] (h10) {};
\node[alpha, below left = .30cm and .30cm of 10] (11) {$\textcolor{blue}{b}$};
\node[round, below right = .15cm and .15cm of 11] (h11) {};
\node[alpha, below left = .30cm and .30cm of 11] (12) {$\textcolor{blue}{b}$};
\node[round, below right = .15cm and .15cm of 12] (h12) {};
\node[alpha, below left = .30cm and .30cm of 12] (13) {$\textcolor{blue}{b}$};
\node[round, below right = .15cm and .15cm of 13] (h13) {};
\node[alpha, below left = .30cm and .30cm of 13] (14) {$\textcolor{red}{b}$};
\node[alpha, below left = .15cm and .15cm of 14] (14a) {$\textcolor{red}{a}$};
\node[alpha, below right = .30cm and .30cm of 14] (15) {$\textcolor{blue}{b}$};
\node[round, below right = .15cm and .15cm of 15] (h15) {};
\node[alpha, below left = .30cm and .30cm of 15] (16) {$\textcolor{blue}{b}$};
\node[round, below right = .15cm and .15cm of 16] (h16) {};
\node[alpha, below left = .15cm and .15cm of 16] (16') {$\textcolor{red}{b}$};
\node[alpha, below left = .15cm and .15cm of 16'] (16'') {$\textcolor{red}{b}$};
\node[alpha, below right = .15cm and .15cm of 16'] (T2') {$T_2$};
\node[alpha, below left = .15cm and .15cm  of 16''] (16a) {$\textcolor{red}{a}$};
\node[alpha, below right = .15cm and .15cm  of 16''] (16aa) {$\textcolor{red}{a}$};
\draw[red, line width = 1.3pt] (1) -- (2) ;
\draw[blue, line width = 1.3pt] (2) -- (3) -- (4) ;
\draw[red, line width = 1.3pt] (4) -- (5) ;
\draw[blue, line width = 1.3pt] (5) -- (6) -- (7) -- (8) -- (9) ;
\draw[line width = 1.3pt] (9) -- (10) ;
\draw[blue, line width = 1.3pt] (10) -- (11) -- (12) -- (13) -- (14);
\draw[red, line width = 1.3pt] (14) -- (15) ;
\draw[blue, line width = 1.3pt] (15) -- (16) -- (16');
\draw[red] (16')-- (16'');
\draw[red] (1) -- (1');
\draw[red] (1') -- (1a);
\draw[red] (1') -- (1aa);
\draw[blue] (2) -- (h2);
\draw[blue] (3) -- (h3);
\draw[blue] (5) -- (h5);
\draw[blue] (6) -- (h6);
\draw[blue] (7) -- (h7);
\draw[blue] (8) -- (h8);
\draw[blue] (10) -- (h10);
\draw[blue] (11) -- (h11);
\draw[blue] (12) -- (h12);
\draw[blue] (13) -- (h13);
\draw[blue] (15) -- (h15);
\draw[blue] (16) -- (h16);
\draw[red] (16') -- (T2');
\draw[red] (16'') -- (16a);
\draw[red] (16'') -- (16aa);
\draw[red] (4) -- (4a);
\draw[red] (14) -- (14a);
\draw (9) -- (T2);

\node[alpha, right = 3cm of 1] (1) {$\textcolor{red}{b}$};
\node[alpha, below left = .15cm and .15cm of 1] (1') {$\textcolor{red}{b}$};
\node[alpha, below left = .15cm and .15cm  of 1'] (1a) {$\textcolor{red}{a}$};
\node[alpha, below right = .15cm and .15cm  of 1'] (1aa) {$\textcolor{red}{a}$};
\node[alpha, below right = .30cm and .30cm of 1] (2) {$\textcolor{blue}{b}$};
\node[alpha, below right = .15cm and .15cm of 2] (h2) {$\textcolor{darkgreen}{a}$};
\node[alpha, below left = .30cm and .30cm of 2] (3) {$\textcolor{blue}{b}$};
\node[alpha, below right = .15cm and .15cm of 3] (h3) {$\textcolor{darkgreen}{a}$};
\node[alpha, below left = .30cm and .30cm of 3] (4) {$\textcolor{red}{b}$};
\node[alpha, below left = .15cm and .15cm  of 4] (4a) {$\textcolor{red}{a}$};
\node[alpha, below right = .30cm and .30cm of 4] (5) {$\textcolor{blue}{b}$};
\node[alpha, below right = .15cm and .15cm of 5] (h5) {$\textcolor{darkgreen}{a}$};
\node[alpha, below left = .30cm and .30cm of 5] (6) {$\textcolor{blue}{b}$};
\node[alpha, below right = .15cm and .15cm of 6] (h6) {$\textcolor{darkgreen}{a}$};
\node[alpha, below left = .30cm and .30cm of 6] (7) {$\textcolor{blue}{b}$};
\node[alpha, below right = .15cm and .15cm of 7] (h7) {$\textcolor{darkgreen}{a}$};
\node[alpha, below left = .30cm and .30cm of 7] (8) {$\textcolor{blue}{b}$};
\node[alpha, below right = .15cm and .15cm of 8] (h8) {$\textcolor{darkgreen}{a}$};
\node[alpha, below left = .30cm and .30cm of 8] (9) {$b$};
\node[alpha, below right = .30cm and .30cm of 9] (10) {$\textcolor{blue}{b}$};
\node[alpha, below left = .15cm and .15cm of 9] (T2) {$T_2$};
\node[alpha, below right = .15cm and .15cm of 10] (h10) {$T_3$};
\node[alpha, below left = .30cm and .30cm of 10] (11) {$\textcolor{blue}{b}$};
\node[alpha, below right = .15cm and .15cm of 11] (h11) {$\textcolor{darkgreen}{a}$};
\node[alpha, below left = .30cm and .30cm of 11] (12) {$\textcolor{blue}{b}$};
\node[alpha, below right = .15cm and .15cm of 12] (h12) {$\textcolor{darkgreen}{a}$};
\node[alpha, below left = .30cm and .30cm of 12] (13) {$\textcolor{blue}{b}$};
\node[alpha, below right = .15cm and .15cm of 13] (h13) {$\textcolor{darkgreen}{a}$};
\node[alpha, below left = .30cm and .30cm of 13] (14) {$\textcolor{red}{b}$};
\node[alpha, below left = .15cm and .15cm of 14] (14a) {$\textcolor{red}{a}$};
\node[alpha, below right = .30cm and .30cm of 14] (15) {$\textcolor{blue}{b}$};
\node[alpha, below right = .15cm and .15cm of 15] (h15) {$\textcolor{darkgreen}{a}$};
\node[alpha, below left = .30cm and .30cm of 15] (16) {$\textcolor{blue}{b}$};
\node[alpha, below right = .15cm and .15cm of 16] (h16) {$\textcolor{darkgreen}{a}$};
\node[alpha, below left = .15cm and .15cm of 16] (16') {$\textcolor{red}{b}$};
\node[alpha, below left = .15cm and .15cm of 16'] (16'') {$\textcolor{red}{b}$};
\node[alpha, below right = .15cm and .15cm of 16'] (T2') {$T_2$};
\node[alpha, below left = .15cm and .15cm  of 16''] (16a) {$\textcolor{red}{a}$};
\node[alpha, below right = .15cm and .15cm  of 16''] (16aa) {$\textcolor{red}{a}$};
\draw[red, line width = 1.3pt] (1) -- (2) ;
\draw[blue, line width = 1.3pt] (2) -- (3) -- (4) ;
\draw[red, line width = 1.3pt] (4) -- (5) ;
\draw[blue, line width = 1.3pt] (5) -- (6) -- (7) -- (8) -- (9) ;
\draw[line width = 1.3pt] (9) -- (10) ;
\draw[blue, line width = 1.3pt] (10) -- (11) -- (12) -- (13) -- (14);
\draw[red, line width = 1.3pt] (14) -- (15) ;
\draw[blue, line width = 1.3pt] (15) -- (16) -- (16');
\draw[red] (1) -- (1');
\draw[red] (1') -- (1a);
\draw[red] (1') -- (1aa);
\draw[blue] (2) -- (h2);
\draw[blue] (3) -- (h3);
\draw[blue] (5) -- (h5);
\draw[blue] (6) -- (h6);
\draw[blue] (7) -- (h7);
\draw[blue] (8) -- (h8);
\draw[blue] (10) -- (h10);
\draw[blue] (11) -- (h11);
\draw[blue] (12) -- (h12);
\draw[blue] (13) -- (h13);
\draw[blue] (15) -- (h15);
\draw[blue] (16) -- (h16);
\draw[red] (16')-- (16'');
\draw[red] (16') -- (T2');
\draw[red] (16'') -- (16a);
\draw[red] (16'') -- (16aa);
\draw[red] (4) -- (4a);
\draw[red] (14) -- (14a);
\draw (9) -- (T2);
\end{tikzpicture}}
\caption{The tree fragment $f(w) \, b  \, T_2  \, f(\bar{w})  \, T_2$ for $w = 000011001111$ on the left and $t(w)$ on the right.} \label{fig-TC0}
\end{figure}
 
 \begin{claimproof}
 Let us look at the example where $w = d(001011) = 000011001111$ has length $2n=12$, which satisfies $|w|_0 \ge n = 6$.
 On the right of Figure~\ref{fig-TC0}, the tree 
 \[
t(w)  =   f(w) \, b  \, T_2  \, f(\bar{w})  \, T_2  \, a^{n-1}  \, T_3  \, a^n 
 =  \textcolor{red}{bbaa} \, \textcolor{blue}{bb} \, \textcolor{red}{ba} \, \textcolor{blue}{bbbb} \, b  \, T_2  \, \textcolor{blue}{bbbb} \, \textcolor{red}{ba} \, \textcolor{blue}{bb} \, \textcolor{red}{bbaa} \, T_2 \,
\textcolor{darkgreen}{a^5}  \, T_3  \, \textcolor{darkgreen}{a^6} 
\]
is shown.
 Note that a string $b^k a^k$ produces a caterpillar tree of depth $k$ branching off from the root to the left and leaving a ``hole'' at the position right below the root. These are the red patterns in Figure~\ref{fig-TC0}. The $b$'s (replacing the $1$'s when applying $f$) yield the blue nodes
 and edges in Figure~\ref{fig-TC0}.
 
 Figure~\ref{fig-TC0} (left) shows the fragment of $t(w)$ that is produced by the prefix $f(w) \, b  \, T_2  \, f(\bar{w})  \, T_2$. 
The green circles represent holes. The first $|w|_1$ holes are produced by $f(w)$ followed by $|w|_0$ holes produced by $b  \, T_2  \, f(\bar{w})  \, T_2$ (note that $|w|_0 +  |w|_1 = 2n$). 
These $2n$ holes are then filled bottom-up by the $2n$ trees from the suffix $a^{n-1}  \, T_3  \, a^n$, which
finally yields $t(w)$.
 All $2n$ holes are filled with $a$ except the $n^{\text{th}}$ hole from bottom, which is filled by $T_3$.
Note that the tree $t(w)$ has a main spine that is highlighted by the thick edges in Figure~\ref{fig-TC0}.
The nodes on this spine are called the spine nodes below.
All subtrees that are attached to the spine have Strahler number at most $2$ except for the unique occurrence of $T_3$.
The crucial observation now is the following:
\begin{itemize}
\item If $|w|_0 \ge n$ then 
$T_3$ is attached to a spine node that is below the spine node to which the upper occurrence of 
$T_2$ is attached (this is the case in Figure~\ref{fig-TC0}). 
This implies $\hs(t(w)) = 3$.
\item If $|w|_0 < n$ then $T_3$ is attached to a spine node that is above the spine node to which the upper occurrence of 
$T_2$ is attached. This implies that $\hs(t(w)) = 4$.  
\end{itemize} 
This proves Claim~\ref{claim-tc0}.
\end{claimproof}
 Claim~\ref{claim-tc0} implies the second statement of the theorem: if $w \in \mathsf{Majority}$ then 
 $|w|_1 > n$, i.e., $|w|_0 = 2n - |w|_1 < n$, and Claim~\ref{claim-tc0} gives $\hs(t(w)) = 4$.
 Similarly, if $|w|_1 \le n$ then $|w|_0 \geq n$ and Claim~\ref{claim-tc0} gives $\hs(t(w)) = 3$.
\end{proof}
It is easy to see that the problem $\mathsf{St}_{\mathsf{term}}^{\ge 2}$ belongs to $\uAC^0$: if $t \in \Bin$ then $\hs(t) \geq 2$ if and only
if the string $t$ contains at least two occurrences of $baa$, which can be tested in $\uAC^0$. We do not know whether $\mathsf{St}_{\mathsf{term}}^{\ge 3}$ still belongs to $\uAC^0$.

For input trees given in pointer representation, we get:

\begin{theorem} \label{thm-main-LS}
$\StNpointer$ and  $\mathsf{St}_{\mathsf{pointer}}^{\ge k}$ for every $k \geq 3$ are $\LS$-complete.
\end{theorem}

\begin{proof}
The  upper bound for $\StNpointer$ follows from Theorem~\ref{thm-main-NC1} and the fact that the pointer representation
of a tree can be transformed in logspace into its term representation.

For the $\LS$-hardness of $\mathsf{St}_{\mathsf{pointer}}^{\ge 3}$ we give a reduction from the line graph accessibility problem.
A (directed) line graph is a directed graph $(V,E)$ with node set $V = \{v_1, v_2, \ldots, v_n\}$ and edge set $E=\{ (v_i, v_{i+1}) : 1 \le i \le n-1\}$.
The input for the line graph accessibility problem is a directed line graph $G=(V,E)$ and two different nodes $u,v \in V$ and it is asked whether there is
a path from $u$ to $v$ in $G$.
The line graph accessibility problem is known to be $\LS$-complete \cite{Etessami97}.

The reduction from the line graph accessibility problem to $\mathsf{St}_{\mathsf{pointer}}^{\ge 3}$ is similar to the reduction from the proof of 
Theorem~\ref{theorem-uTC0} but less technical.
Fix a line graph $G=(V,E)$ and two nodes $u, v \in V$ with $u \neq v$.  
Let $w \in V$ be the unique node of out-degree $0$ in $G$.
By adding a node to $V$ we can assume that $u \neq w \neq v$. 
From the line graph $G$ we construct a binary tree by adding first a single new child to every
node in $V \setminus \{ u, v, w\}$ and two children to $w$ (the new children all leaves in the tree).
 Finally, we add a binary tree of Strahler number $2$ (resp., $1$) whose root will be the second
child of $u$ (resp., $v$). Let $t$ be the resulting binary tree. If there is a path from $u$ to $v$ in 
$G$ then $\hs(t) = 3$, otherwise $\hs(t) = 2$.
\end{proof}
Finally, we also considered the cases where the input tree is given in a compressed form by either a binary DAG or a TSLP.

\begin{theorem} \label{thm-compressed} The following hold:
\begin{enumerate}[(i)]
\item $\StNdag$ and $\StNtslp$ are \P-complete.
\item For every fixed $k \geq 3$,  $\StNkdag$ is $\LS$-complete. 
\item For every fixed $k \geq 2$,  $\StNktslp$ is $\NL$-complete. 
\end{enumerate}
\end{theorem}

\begin{proof}[Proof of Theorem~\ref{thm-compressed}(i)]
It suffices to show the lower bound for DAGs and the upper bound for TSLPs.
For the lower bound for DAGs one can reuse the reduction from the Boolean formula value problem
in the proof of Theorem~\ref{thm-main-NC1} in order to reduce the $\P$-complete monotone Boolean circuit value 
problem to $\StNdag$.

For the upper bound for TSLPs let $\mathcal{G} = (N_0, N_1, S, \rho)$ be a TSLP. Then we can compute 
bottom-up for every $A \in N_0$ the value $\hs_A$ (see the paragraph after Theorem~\ref{theorem tree->tslp})
and for every $B \in N_1$ the values $\ell_B$ and $h_B$ 
from Lemma~\ref{lemma-l_B-h_B}. All these values are bounded by 
$\log_2 n$, where $n$ is the number of leaves of the tree $\val(\mathcal{G})$
(this was shown in the proof of Theorem~\ref{thm-main-NC1}).  Hence, all values are bounded
by $\mathcal{O}(|\mathcal{G}|)$ and can be computed bottom-up in polynomial time according to the rules \eqref{evaluate-1}--\eqref{evaluate-5}
from the proof of Theorem~\ref{thm-main-NC1}.
\end{proof}

\begin{algorithm}
\caption{Recursive Strahler number computation}
\label{algo:rec}
\begin{algorithmic}[1]
    \Procedure{ST}{$v, k$}
        \If{$v$ is a leaf}
            \State \Return \textbf{true}
        \EndIf
        
        \State $v_1 v_2 \gets \gamma(v)$ %
        \State $b_1 \gets \text{ST}(v_1, k - 1)$  \label{rec1} 
        \State $b_2 \gets \text{ST}(v_2, k - 1)$ \label{rec2}
        
        \If{$b_1 = b_2$}
            \State \Return $b_1$ \label{case1}
        \ElsIf{$b_1 = 1 \land b_2 = 0$}
            \State \Return $\text{ST}(v_2, k)$ \label{case2}
        \ElsIf{$b_1 = 0 \land b_2 = 1$}
            \State \Return $\text{ST}(v_1, k)$ \label{case3}
        \EndIf
    \EndProcedure
\end{algorithmic}
\end{algorithm}

\begin{proof}[Proof of Theorem~\ref{thm-compressed}(ii)]
Let $\mathcal{D} = (V, v_0,\gamma)$ be the binary DAG from the input.
\cref{algo:rec} is a straightforward recursive algorithm that tests whether the Strahler number of a node $v \in V$ is at most $k$. It returns 
$1$ if $\hs(v) \le k$ and $0$ if $\hs(v) > k$.
The correctness of the algorithm follows immediately from the fact that
\begin{align*}
	\hs(v) \le k \iff \quad & (\hs(v_1) \le k-1 \wedge \hs(v_2) \le k-1) \\
	\vee ~ & (\hs(v_1) \le k-1 \wedge \hs(v_2) = k) \\
	\vee ~ & (\hs(v_1) = k \wedge \hs(v_2) \le k-1) ,
\end{align*}
where the three disjuncts correspond to the cases in lines \ref{case1}, \ref{case2}, and \ref{case3}.
Furthermore, we claim that the algorithm can be implemented in $O(k \cdot \log n)$ space
where $n$ is the input size.
Observe that the recursive calls in lines \ref{rec1} and \ref{rec2} decrease the parameter $k$ by one.
Furthermore, the recursive calls in lines \ref{case2} and \ref{case3} are tail recursive calls.
Therefore the entire procedure can be implemented using a stack of $k$ many $O(\log n)$-bit entries, each containing a constant number
of nodes and bits.
The $\LS$-hardness for $k \geq 3$ comes from Theorem~\ref{thm-main-LS}.
\end{proof}

\begin{proof}[Proof of Theorem~\ref{thm-compressed}(iii)]
We show that $\StNktslp$ can be accepted by an alternating logspace machine with $k$ alternations.
Since $k$ is a fixed constant, the Immerman--Szelepcsényi
theorem then implies that $\StNktslp$ belongs to $\NL$ \cite[Corollary~2]{Imm88}. 

Let $\mathcal{G} = (N_0, N_1, S, \rho)$ be the input TSLP. We have to check whether 
$\hs(\val(\mathcal{G})) \le k$ for a fixed value $k$. 
Recall the notation $\hs_A$ for a variable $A$; see the paragraph after 
Theorem~\ref{theorem tree->tslp}. W.l.o.g.~we can replace every $\rho(A) = b(x,B)$ by $\rho(A) = b(B,x)$ (this does not 
change the Strahler number of $\val(\mathcal{G})$).

For every $A \in N_0$ such that $\rho(A) = C(B)$ for some $C \in N_1$ and $B \in N_0$
we define a binary DAG 
$\mathcal{D}_A = (N, A, \gamma)$ (without a node-labelling function $\lambda$),
where $\gamma$ is defined as follows:
\begin{itemize}
\item $\gamma(A) = CB$,
\item $\gamma(D) = EF$ if $D \in N_1$ and $\rho(D) = E(F(x))$,
\item $\gamma(D) = E$ if $D \in N_1$ and $\rho(D) = b(E,x)$,
\item $\gamma(D) = \varepsilon$ in all other cases.
\end{itemize}

\begin{claim} \label{claim-NL} 
For every $m \geq 1$ and $A \in N_0$ such that $\rho(A) = C(B)$ ($C \in N_1$ and $B \in N_0$)
we have $\hs_A \geq m$ if and only if one of the following cases holds:
\begin{itemize}
\item In the DAG $\mathcal{D}_A$ there exists a path from $A$ to a leaf $D \in N_0$ such that $\hs_D \geq m$.
\item In the DAG $\mathcal{D}_A$ there exist two different paths  ending in leaves $D_1 \in N_0$ and $D_2 \in N_0$, respectively,
such that $\hs_{D_1} \geq m-1$ and $\hs_{D_2} \geq m-1$.
\end{itemize} 
\end{claim}

 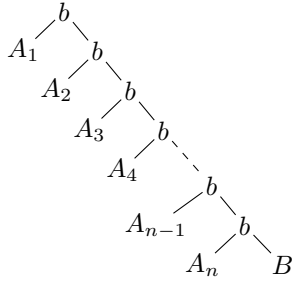
\begin{figure}[t]
 \tikzset{alpha/.style={inner sep = 1pt, fill=white}}
 \tikzset{alpha2/.style={inner sep = 0pt, fill=white}}
\centering{
\begin{tikzpicture}
\node[alpha] (1) {$b$};
\node[alpha2, below left = .2cm and .2cm of 1] (1') {$A_1$};
\node[alpha, below right = .2cm and .2cm of 1] (2) {$b$};
\node[alpha2, below left = .2cm and .2cm of 2] (2') {$A_2$};
\node[alpha, below right = .2cm and .2cm of 2] (3) {$b$};
\node[alpha2, below left = .2cm and .2cm of 3] (3') {$A_3$};
\node[alpha, below right = .2cm and .2cm of 3] (4) {$b$};
\node[alpha2, below left = .2cm and .2cm of 4] (4') {$A_4$};
\node[alpha, below right = .4cm and .4cm of 4] (5) {$b$};
\node[alpha2, below left = .2cm and .2cm of 5] (5') {$A_{n-1}$};
\node[alpha, below right = .2cm and .2cm of 5] (6) {$b$};
\node[alpha2, below left = .2cm and .2cm of 6] (6') {$A_n$};
\node[alpha, below right = .2cm and .2cm of 6] (7) {$B$};
\draw (1) -- (2) -- (3) -- (4) ;
\draw[dashed] (4) -- (5) ;
\draw (5) -- (6) ;
\draw (1) -- (1') ;
\draw (2) -- (2') ;
\draw (3) -- (3') ;
\draw (4) -- (4') ;
\draw (5) -- (5') ;
\draw (6) -- (6') ;
\draw (6) -- (7) ;
 \end{tikzpicture} }
\caption{\label{fig-caterpillar} A caterpillar tree from Claim~\ref{claim-NL}.}
\end{figure}

\begin{claimproof}
Note that from $A$ we can derive a caterpillar tree as shown in Figure~\ref{fig-caterpillar}. 
To simplify notation, let us write $A_{n+1}$ for $B$. All $A_i$ ($0 \le i \le n+1$) in 
Figure~\ref{fig-caterpillar} are from $N_0$ and are leaves of the DAG $\mathcal{D}_A$.
 Then $\hs(A) \ge m$ if and only if (i) there 
is an $1 \le i \le n+1$ such that $\hs(A_i) \ge m$ or (ii) there are $1 \le i < j \le n+1$ such that $\hs(A_i) \ge m-1$ 
and  $\hs(A_j) \ge m-1$. These two cases (i) and (ii) correspond to the two cases from the claim.
\end{claimproof}

\noindent
Claim~\ref{claim-NL}  leads to the following alternating logspace algorithm: 
As in the proof of Theorem~\ref{thm-compressed}(ii), the algorithm stores an active statement $\mathcal{S}_a = [\hs_A \geq m]$
for some $A \in N_0$ and $0 \le m \le k$ (we do not need the stack used in the proof of Theorem~\ref{thm-compressed}(ii)).
Initially, it sets $\mathcal{S}_a := [\hs_S \ge k]$. If at some point we have $\mathcal{S}_a = [\hs_A \geq 0]$ for some
$A \in N_0$ then the algorithm accepts. If $\mathcal{S}_a = [\hs_A \geq m]$ for some $m>0$ and 
$\rho(A) = a$ then the algorithm rejects.

Now assume that $\mathcal{S}_a = [\hs_A \geq m]$ with $m > 0$, $A \in N_0$ and $\rho(A) \neq a$. Then we distinguish the following two cases:

\medskip
\noindent
\emph{Case 1.} $\rho(A) = C(B)$ (and hence $\gamma(A) = CB$ in the DAG $\mathcal{D}_A$)
 for some $C \in N_1$ and $B \in N_0$. Then the algorithm nondeterministically chooses one of the following two branches:
\begin{itemize} 
\item guess existentially a path in $\mathcal{D}_A$ from $A$ to a leaf $D$, set $\mathcal{S}_a := [\hs_D \geq m]$ and continue, 
\item guess existentially a (possibly empty) path in $\mathcal{D}_A$ from $A$ to a node $D$ such that $\gamma(D) = D_1 D_2$ for some 
$D_1, D_2$. Then guess universally an $i \in \{1,2\}$. Finally, guess existentially a path in $\mathcal{D}_A$ from $D_i$ to a leaf $E \in N_0$,
set $\mathcal{S}_a := [\hs_E \geq m-1]$ and continue.
\end{itemize}
Note that the first (resp., second) continuation corresponds to the first (resp., second) point in Claim~\ref{claim-NL}.
 
\medskip
\noindent
\emph{Case 2.} $\rho(A) = b(A_1, A_2)$  for some $A_1 , A_2 \in N_0$.
Then the algorithm guesses existentially one of the following three branches:
\begin{itemize} 
\item set $\mathcal{S}_a = [\hs_{A_1} \geq m]$ and continue,
\item set $\mathcal{S}_a = [\hs_{A_2} \geq m]$ and continue,
\item universally guess an $i \in \{1,2\}$ and set $\mathcal{S}_a = [\hs_{A_i} \geq m-1]$.
\end{itemize}
This algorithm is clearly correct. Moreover, the number of alternations is bounded by $k$ since after each universal 
guess the value $\num(\mathcal{S}_a)$ (defined as in the proof of Theorem~\ref{thm-compressed}(ii)) is decremented and the algorithm stops when $\num(\mathcal{S}_a)=0$ (or earlier).

We show $\NL$-hardness for $k=2$ by a reduction from the following variant of the graph accessibility problem for DAGs.
Consider a binary DAG $\mathcal{D} = (V, v_0, \gamma)$ as defined in Section~\ref{sec-dags-trees} 
(but without the labelling function $\lambda$), where
$d(v) \in \{0,2\}$ for every node $v \in V$ and $d(v_0)=2$. The question is, whether for a given target node $v_t$ 
with $\gamma(v_t) = \varepsilon$, there
is a path from $v_0$ to $v_t$. The standard graph accessibility problem for DAGs can be easily reduced to our variant.
We construct a TSLP $\mathcal{G} = (N_0, N_1, S, \rho)$ as follows:
\begin{itemize}
\item $N_0 = \{ S, A, B \}$, $N_1 = V$ (of course we assume $N_0 \cap N_1 = \emptyset$),
\item $\rho(S) = v_0(B)$,
\item $\rho(v) =  v_1(v_2(x))$ if $v \in V$ and $\gamma(v) = v_1 v_2$,
\item $\rho(v) = b(x, A)$ if $v \in V$, $\gamma(v) = \varepsilon$ and $v \neq v_t$,
\item $\rho(v) = b(x, B)$ if  $v = v_t$,
\item $\rho(A) = a$, $\rho(B) = b(A,A)$.
\end{itemize}
If there is no path from $v_0$ to $v_t$ then $\val(\mathcal{G})$ is a caterpillar tree of the form
$b^m a^{m+1}$ for some $m$ and hence $\hs(\val(\mathcal{G})) = 1$. On the other hand, if 
there is a path from $v_0$ to $v_t$ then $\val(\mathcal{G})$ has the form $b^m a^i b a^j$ 
with $i,j \ge 2$ and $i+j = m+2$ and hence $\hs(\val(\mathcal{G})) = 2$.
\end{proof}

\section{Strahler number of derivation trees} \label{sec-CNF}

In this section, we consider the Strahler numbers of derivation trees of 
context-free grammars. 
Since we want to obtain binary trees and since the derived words have no relevance for us,
we consider context-free grammars $G = (N,S,P)$, where $N$ is the set of nonterminals, $S \in N$ is the start nonterminal
and $P$ is the set of productions such that each of them has the form $A \to \varepsilon$ or $A \to BC$ for $A,B,C \in N$.
Slightly abusing standard terminology, we call such a grammar a \emph{Chomsky normal form grammar} or CNF-grammar for short.
The notion of a derivation tree is defined as usual: a \emph{derivation tree for} $A \in N$ 
is an $N$-labelled binary tree such that (i) the root is labelled with $A$,
(ii) if an internal  node $v$ is labelled with $B \in N$
then there is a production $B \to CD$ such that the left (resp., right) child of $v$ is labelled with $C$ (resp., $D$) and (iii)
if $v$ is a $B$-labelled leaf then $(B \to \varepsilon) \in P$. A \emph{derivation tree of} $G$ is a derivation for the start nonterminal $S$.
The grammar $G$ is \emph{productive} if for every nonterminal $A \in N$ there is a derivation tree.
It is well-known that a given CNF-grammar can be transformed in polynomial time into an equivalent 
productive CNF-grammar. 
A derivation tree $t$ is called acyclic if there is no nonterminal that appears twice
along a path from the root to a leaf. The motivation for considering acyclic derivations trees and their Strahler numbers comes
from \cite{LRZ25}; see the discussion in the introduction.

In this section we consider the following problem $\CNF^\ge$ (resp., $\acCNF^\ge$):
\begin{itemize}
\item Input: a CNF-grammar $G$ and a number $k$ (given in unary encoding).
\item Question:  Is there a derivation tree (resp., acyclic derivation tree)  $t$ of $G$ with $\hs(t) \geq k$? 
\end{itemize}
If the number $k$ is fixed and not part of the input, we obtain the problems $\CNF^{\ge k}$ and $\acCNF^{\ge k}$.
The following results pinpoint the complexity of these problems.

\begin{theorem} \label{thm-CNF}
The following holds:
\begin{enumerate}[(i)]
\item
$\CNF^\ge$ and $\CNF^{\ge k}$ for every $k \geq 1$ are $\P$-complete.
\item $\CNFack$ is $\NP$-complete for every $k \geq 2$.
\item $\CNFac$ is $\PSPACE$-complete.
\end{enumerate}
\end{theorem}

\begin{proof}[Proof of Theorem~\ref{thm-CNF}(i)]
For the $\P$ upper bound for $\CNF^\ge$, let  $G = (N,S,P)$ be a CNF-grammar. W.l.o.g.~we can assume that $G$ is productive.
For every nonterminal $A \in N$ let
\[ \hs_A = \mathsf{sup} \{ \hs(t) : \text{ $t$ is a derivation tree for $A$}\},\] 
where $\mathsf{sup}$ refers to the supremum and $\mathsf{sup}(M) = \infty$ for an unbounded set $M \subseteq \mathbb{N}$.
It suffices to compute in polynomial time the value $\hs_S$.

First assume that there exists a nonterminal $A \in N$
such that $A \derive_G^* AA$, i.e., $AA$ can be derived with the productions from $A$.
Since $G$ is productive this implies $\hs_A = \infty$ and hence $\hs_S = \infty$.

Now assume that there is no $A \in N$
such that $A \derive_G^* AA$.
We claim that $\hs_S \leq |N|$ (which implies that $\hs_A \leq |N|$ for all $A \in N$).
In order to get a contradiction, assume that $G$ has a derivation tree $t$ with $h := \hs(t) > |N|$.
Hence, there exists a path from the root of $t$ to a node $u_0$ such that its two children $v_{0}$ and $w_{0}$ satisfy
$\hs(t(v_{0})) = h-1$ and $\hs(t(w_{0})) = h-1$. Let $A_0$ be the nonterminal labelling
$u_0$. By assumption, $A_0$ cannot
occur in both $t(v_{0})$ and $t(w_{0})$. W.l.o.g. assume that $A_0$ does not occur in $t(v_{0})$.
We then repeat this argument with $t(v_{0})$. In this way we obtain nodes $u_0, v_0, u_1, v_1, \ldots, u_{h-1}, v_{h-1},u_h$ such that
$v_i$ is a child of $u_i$, $u_{i+1}$ is a descendant of $v_i$ and, if $A_i$ is the nonterminal labelling $u_i$, then 
$A_0, \ldots, A_{i}$ do not occur in the subtree $t(v_i)$ and hence do not occur in $t(u_{i+1})$.  
In particular, $A_0, \ldots, A_{h}$ are pairwise different, which contradicts $h > |N|$.

Now, that we know that $\hs_A \leq |N|$ for every $A \in N$, we can easily compute the exact values $\hs_A$
 by a simple fixpoint iteration process.
Initially we set $\hs_A :=0$ for every $A \in N$. Then, as long as there is a production $(A \to BC) \in P$ with $A,B,C \in N$ and
$\hs_A < s(\hs_B, \hs_C)$, we set $\hs_A := s(\hs_B, \hs_C)$. After at most $|N|^2$ many steps we must reach a fixpoint
(for each of the $|N|$ many nonterminals $A \in N$ the value $\hs_A$ can only increase $|N|$ times).

$\P$-hardness of $\CNF^{\ge k}$ for $k \ge 1$ can be shown by a straightforward
reduction from the emptiness problem for CNF-grammars.
\end{proof}

\begin{proof}[Proof of Theorem~\ref{thm-CNF}(ii)]
Membership in $\NP$ was shown in \cite{LRZ25}. 
It suffices to show $\NP$-hardness for $k = 2$. For this, we present a reduction from
\emph{exact 3-hitting set} (X3HS):
\begin{itemize}
\item Input: a finite set $M$ and a non-empty set $\mathcal{B} \subseteq 2^M$ of subsets of $M$, all of size 3.
\item Question: Is there a subset $S \subseteq M$ such that $|S \cap C|=1$ for all $C \in \mathcal{B}$? 
\end{itemize}
X3HS is the same problem as positive 1-in-3-SAT, which is $\NP$-complete \cite[Problem LO4]{GarJoh79}. 

Fix the set $M$ and a subset $\mathcal{B} \subseteq 2^M$ with all $C \in \mathcal{B}$ of size $3$. W.l.o.g.\ assume that
$M  = \{1,2,\ldots, n\}$ and fix an arbitrary ordering $C_1, C_2, \ldots, C_m$ of the subsets in $\mathcal{B}$.
We will construct a CNF-grammar $G$ such that there is a derivation tree of $G$ with Strahler number at least two if and only if
 there is a subset $S \subseteq \{1,\ldots,n\}$ such that $|S \cap C|=1$ for all $C \in \mathcal{B}$.
 
 In order to make the grammar more readable we use the following notation below. If we write in a right-hand side
 $[AB]$ \label{page-[AB]} for nonterminals $A$ and $B$, then $[AB]$ is another nonterminal with the unique production $[AB] \to AB$ and this
 production is not explicitly listed. Moreover, this notation will be nested, i.e., $A$ and $B$ can be also of the form $[CD]$.
 In Figure~\ref{fig coNP} such nonterminals are depicted as filled circles. 
 With this notation, the productions of our CNF-grammar $G$ are as follows: 
  \begin{eqnarray*}
 E & \to & \varepsilon \\
 A_k & \to & I_k E \mid O_k E \text{ whenever $1 \leq k \leq n$} \\
 I_k & \to & A_{k+1} E \mid \varepsilon   \text{ whenever $1 \leq k \leq n-1$} \\
  I_n & \to & B_1 E \mid \varepsilon  \\
 O_k & \to & A_{k+1} E \mid \varepsilon   \text{ whenever $1 \leq k \leq n-1$} \\
 O_n & \to & B_1 E \mid \varepsilon  \\
 B_j & \to & B_{j+1} [[O_a [I_b I_c]] \text{ whenever $1 \leq j \leq m$ and $C_j = \{a,b,c\}$} \\
 B_{m+1} & \to & EE 
 \end{eqnarray*}
 The start nonterminal is $A_1$. Note that there are six productions of the form $B_j  \to  B_{j+1} [[O_a [I_b I_c]]$
 corresponding to the six permutations of the set $C_j = \{a,b,c\}$ (we could restrict to three productions since the order
 between $I_b$ and $I_c$ is not important for the following arguments).
 
 \begin{figure}[t]
 \tikzset{alpha/.style={inner sep = 1pt, fill=white}}
  \tikzset{round/.style={inner sep = 1.5pt, circle, fill=black}}
\centering{
\begin{tikzpicture}

 \begin{scope}[scale=.8]   
\draw (3,0) node[alpha] (A1) {$A_1$} 
           -- ++(-.4,-1)  node[alpha]  (X1) {$X_1$}
           -- ++(-.4,-1)  node[alpha]  (A2) {$A_2$} 
           -- ++(-.4,-1)  node[alpha]  (X2) {$X_2$} 
           -- ++(-.4,-1)  ;
           
\draw[dashed] (1.4,-4) node[alpha] (A3) {$A_3$} -- ++(-.8,-2);

\draw (0.6,-6) node[alpha] (Xi-1) {$\ \ \  X_{i-1}$}
           -- ++(-.4,-1)  node[alpha] (Ai) {$A_i$} 
           -- ++(-.4,-1)  node[alpha]  (Xi) {$X_i$};

\draw (A1) -- ++(.4,-1) node[alpha]  {$E$} ;
\draw (X1) -- ++(.4,-1) node[alpha]  {$E$} ;
\draw (A2) -- ++(.4,-1) node[alpha]  {$E$} ;
\draw (X2) -- ++(.4,-1) node[alpha]  {$E$} ;
\draw (A3) -- ++(.4,-1) node[alpha]  {$E$} ;
\draw (Xi-1) -- ++(.4,-1) node[alpha]  {$E$} ;
\draw (Ai) -- ++(.4,-1) node[alpha]  {$E$} ;

\draw (A1) ++(3,0) node[alpha] (A1') {$A_1$} 
           -- ++(-.4,-1)  node[alpha]  (X1') {$X_1$}
           -- ++(-.4,-1)  node[alpha]  (A2') {$A_2$} 
           -- ++(-.4,-1)  node[alpha]  (X2') {$X_2$} 
           -- ++(-.4,-1)  ;
           
\draw[dashed] (4.4,-4) node[alpha] (A3') {$A_3$} -- ++(-.8,-2);

\draw (3.6,-6) node[alpha] (Xn-1) {$\ \ \ \  X_{n-1}$}
           -- ++(-.4,-1)  node[alpha] (An) {$A_n$} 
           -- ++(-.4,-1)  node[alpha]  (Xn) {$X_n$}
           -- ++(-.4,-1)  node[alpha]  (B1) {$B_1$};
          
\draw (A1') -- ++(.4,-1) node[alpha]  {$E$} ;
\draw (X1') -- ++(.4,-1) node[alpha]  {$E$} ;
\draw (A2') -- ++(.4,-1) node[alpha]  {$E$} ;
\draw (X2') -- ++(.4,-1) node[alpha]  {$E$} ;
\draw (A3') -- ++(.4,-1) node[alpha]  {$E$} ;
\draw (Xn-1) -- ++(.4,-1) node[alpha]  {$E$} ;
\draw (An) -- ++(.4,-1) node[alpha]  {$E$} ;
\draw (Xn) -- ++(.4,-1) node[alpha]  {$E$} ;

\draw (A1') ++(3,0) node[alpha] (B1) {\small $B_1$} 
          -- ++(-.4,-1)  node[alpha]  (B2) {\small $B_2$}
          -- ++(-.4,-1) ;
           
\draw[dashed] (8.2,-2) node[alpha] (B3) {\small $B_3$} -- ++(-.8,-2);

\draw (7.4,-4) node[alpha] (Bm-1) {\small $\ \ \ \  B_{m-1}$}
           -- ++(-.4,-1)  node[alpha] (Bm) {\small $B_m$}
           -- ++(-.4,-1)  node[alpha] (C) {\small $\ \ \ \  B_{m+1}$}
           -- ++(-.4,-1)  node[alpha] (D1) {$E$} ;
    \end{scope} 
 
 \begin{scope}[scale=1] 
 \draw (B1) -- ++(.7,-1) node[round] (B1') {} ;
 \draw (B2) -- ++(.7,-1) node[round] (B2') {} ;
 \draw (B3) -- ++(.7,-1) node[round] (B3') {} ;
 \draw (Bm-1) -- ++(.7,-1) node[round] (Bm-1') {} ;
 \draw (Bm) -- ++(.7,-1) node[round] (Bm') {} ;
 \draw (C) -- ++(.7,-1) node[alpha] (D2) {$E$} ;
 
 \draw (B1') -- ++(.35,-.9) node[alpha] {\small $\ O_{x_1}$} ;
 \draw (B1') -- ++(1.2,-1) node[round] (B1'') {} ; 
 \draw (B1'') -- ++(.1,-.8) node[alpha] {\small $I_{y_1}\! \! \! \!$} ;
 \draw (B1'') -- ++(1.2,-.8) node[alpha] {\small $\! I_{z_1}$} ;
 
 \draw (B2') -- ++(.35,-.9) node[alpha] {\small $\ O_{x_2}$};
 \draw (B2') -- ++(1.2,-1) node[round] (B2'') {} ; 
 \draw (B2'') -- ++(.1,-.8) node[alpha] {\small $I_{y_2}\! \! \! \!$} ;
 \draw (B2'') -- ++(1.2,-.8) node[alpha] {\small $\! I_{z_2}$} ;
 
 \draw (B3') -- ++(.35,-.9) node[alpha] {\small $\ O_{x_3}$};
 \draw (B3') -- ++(1.2,-1) node[round] (B3'') {} ; 
 \draw (B3'') -- ++(.1,-.8) node[alpha] {\small $I_{y_3}\! \! \! \!$} ;
 \draw (B3'') -- ++(1.2,-.8) node[alpha] {\small $\! I_{z_3}$} ;
 
 \draw (Bm-1') -- ++(.35,-1) node[alpha] {\small $O_{x_{m-1}}  \! \! \!\! \! \!\!\!\!\!$};
 \draw (Bm-1') -- ++(1.2,-1) node[round] (Bm-1'') {} ; 
 \draw (Bm-1'') -- ++(.1,-.8) node[alpha] {\small $I_{y_{m-1}} \! \! \!\! \! \!\!\!\!\!$} ;
 \draw (Bm-1'') -- ++(1.2,-.8) node[alpha] {\small $I_{z_{m-1}}  \! \! \! $} ;
 
 \draw (Bm') -- ++(.35,-1) node[alpha] {\small $O_{x_{m}} \! \! \!\! \!$};
 \draw (Bm') -- ++(1.2,-1) node[round] (Bm'') {} ; 
 \draw (Bm'') -- ++(.1,-.8) node[alpha] {\small $I_{y_{m}} \! \! \! \!$} ;
 \draw (Bm'') -- ++(1.2,-.8) node[alpha] {\small $I_{z_{m}}$} ;
 \end{scope}
 \end{tikzpicture} }
\caption{\label{fig coNP} An acyclic derivation tree for the grammar $G$ (proof of Theorem~\ref{thm-CNF}(ii)) has either the form 
shown on the left, or it results from merging the tree shown in the middle with the tree shown on the right in the $B_1$-labelled node.
Every $X_k$ is either $I_k$ or $O_k$.}
\end{figure}
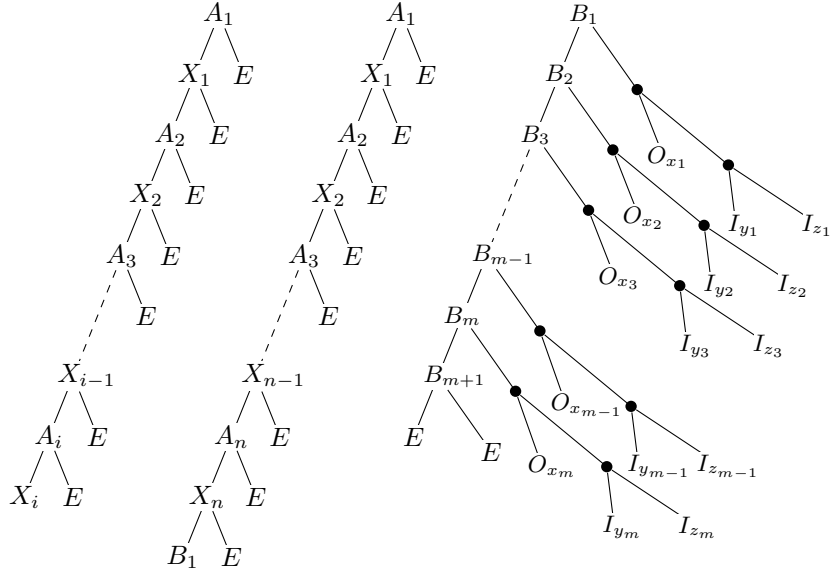

 Consider first an acyclic derivation tree $t$ rooted in $A_1$ with Strahler number at least two. The top part of every derivation tree rooted
 in $A_1$ must have one of the two shapes shown in Figure~\ref{fig coNP} (left and middle tree), where
 $X_k \in \{ I_k, O_k \}$.
 A left tree is a complete (acyclic) derivation tree with Strahler number $1$. Hence, the top part of $t$ must have the middle shape from 
  Figure~\ref{fig coNP}. It defines the
subset $S = \{ k : 1 \leq k \leq n, X_k = I_k \}$. 
 From the leaf $B_1$ we have to expand the derivation tree $t$ further and this results in a bottom part tree as shown in Figure~\ref{fig coNP} (right tree), where for every $1 \le j \le m$ we have $C_j = \{ x_j, y_j, z_j\}$.
Since the tree $t$ (obtained by merging the top part from Figure ~\ref{fig coNP} (middle tree) with the 
bottom part from Figure ~\ref{fig coNP} (right tree), where the merging is done by identifying the $B_1$-labelled nodes)
 is an acyclic derivation tree we must have $x_m \in S$, $y_m \notin S$, and $z_m \notin S$ for every $j \in \{1, \ldots, m\}$.
 Therefore, our X3HS-instance is positive. 
 
 Vice versa, if there is a subset $S \subseteq \{1,\ldots,n\}$ such that $|S \cap C|=1$ for every $C \in \mathcal{B}$, then we obtain an acyclic 
 derivation tree $t$ with Strahler number two by merging the middle and right tree from Figure~\ref{fig coNP}, where we set
 $X_k = I_k$ if $k \in S$ and $X_k = O_k$ if $k \notin S$ in the middle tree. 
 Moreover, if $C_j = \{a,b,c\}$ with $\{a\} = S \cap C$ then we set
 $x_j = a$, $y_j = b$ and $z_j = c$ (or $y_j = c$ and $z_j = b$) in the right tree.
 \end{proof}
 
 \begin{proof}[Proof of Theorem~\ref{thm-CNF}(iii)]
 We first show that $\CNFac$ is in $\PSPACE$. Let $G = (N,S, P)$ be a CNF-grammar and $k \in \mathbb{N}$.
 W.l.o.g.\ we can assume that $G$ is productive.
  We have to check whether $G$ produces an acyclic derivation tree $t$ with $\hs(t) \geq k$. 
  We devise an alternating polynomial time algorithm for this (recall that $\PSPACE$ is equal to alternating polynomial time).
  
   Let $T$ be the set of all triples $(A, U, i)$ such that $A \in N$, $A \in U \subseteq N$ and $0 \le i \le k$.
 The algorithm stores a triple $\tau \in T$, which is initially set to $(S, \{S\}, k)$.
  A triple $(A,U,i) \in T$ stands for the goal of finding
 an acyclic derivation tree $t$ with root $A$ and $\hs(t) \geq i$ such that in addition none of the 
 nonterminals in $U$ appears in $t$ except for the $A$ at the root.
 
 Assume that $\tau = (A, U, i)$. There are three cases:
 
 \medskip\noindent
 \emph{Case 1:}  $i = 0$ and $(A \to \varepsilon) \in P$. Then the algorithm accepts.
 
 \medskip\noindent
 \emph{Case 2:} $i = 0$ and $(A \to \varepsilon) \notin P$. The algorithm guesses nondeterministically a production $A \to A_1 A_2$  such that $A_1 \notin U$ and $A_2 \notin U$ (if such a production does not exist then it rejects). 
 Then it universally guesses an $i \in \{1,2\}$, sets $\tau := (A_i, U \cup \{A_i\}, 0)$ and continues.
 
 \medskip\noindent
 \emph{Case 3:} $i > 0$. The algorithm guesses nondeterministically a production $A \to A_1 A_2$  such that $A_1 \notin U$ and $A_2 \notin U$
 (if such a production does not exist then it rejects).  It then branches existentially to one of the following three continuations:
 \begin{itemize}
\item universally branch to $\tau := (A_1,  U \cup \{A_1\}, i)$ or $\tau := (A_2,  U \cup \{A_2\}, 0)$,
\item universally branch to $\tau := (A_2,  U \cup \{A_2\}, i)$ or $\tau := (A_1,  U \cup \{A_1\}, 0)$,
\item universally branch to $\tau := (A_1,  U \cup \{A_1\}, i-1)$ or $\tau := (A_2,  U \cup \{A_2\}, i-1)$.
\end{itemize}
In each case, the algorithm continues after resetting $\tau$.

It is easy to see that this algorithm is correct. It runs in polynomial time, since in every step the size of the set $U$ in the current triple $\tau$ gets larger.

Let us now come to $\PSPACE$-hardness. We show a reduction from \textsf{QBF}, i.e., the problem whether a quantified Boolean formula
is true. Consider a quantified Boolean formula $\psi = Q_1 x_1 Q_2 x_2 \cdots Q_n x_n \phi$, where
$Q_i \in \{ \exists, \forall \}$ for all $1 \le i \le n$ and  $\phi = \phi(x_1, \ldots, x_n)$ is a Boolean formula containing the variables
$x_1, \ldots, x_n$.
We construct a CNF-grammar $G$ and a number $k$ such that $\psi$ is true if and only if $G$ has an acyclic derivation tree of Strahler number
at least $k$. For the proof we reuse the construction from the proof of Theorems~\ref{thm-main-NC1} 
($\uNC^1$-hardness).
The CNF-grammar $G$ contains the following productions for all $1 \le i \leq n$, where $A_1$ is the start nonterminal of $G$
(as in the previous proof of the $\NP$-hardness for $\CNFack$ 
we use new nonterminals of the form $[AB]$ for nonterminals $A$ and $B$):
\begin{eqnarray}
E & \to & \varepsilon \nonumber \\
A_i & \to & 
\begin{cases}
[F_i F_i ] [T_i T_i]                                                           
& \text{ if } Q_i = \forall \\[2mm]
\big[[F_i F_i] T_i \big] \big[ F_i [T_i T_i]  \big] 
& \text{ if } Q_i = \exists
\end{cases} \label{rule-A_i} \\
F_i & \to & A_{i+1} E \mid \varepsilon \nonumber \\
T_i & \to & A_{i+1} E \mid \varepsilon \nonumber
\end{eqnarray}
\begin{figure}[t]
 \tikzset{alpha/.style={inner sep = 0pt, fill=white}}
  \tikzset{round/.style={inner sep = 1.2pt, circle, fill=black}}
\centering{
\begin{tikzpicture}
\node[alpha] (A) {$A_i$};
\node[round, below left = .4cm and .8cm of A] (A1) {};
\node[round, below right = .4cm and .8cm  of A] (A2) {};
\node[alpha, below left = .2cm and .4cm of A1] (A11) {\small $F_i$};
\node[alpha, below right = .2cm and .4cm  of A1] (A12) {\small $F_i$};
\node[alpha, below left = .2cm and .4cm of A2] (A21) {\small $T_i$};
\node[alpha, below right = .2cm and .4cm  of A2] (A22) {\small $T_i$};

\draw (A) -- (A1);
\draw (A) -- (A2);
\draw (A1) -- (A11);
\draw (A1) -- (A12);
\draw (A2) -- (A21);
\draw (A2) -- (A22);

\node[alpha, right = 5cm of A] (A) {$A_i$};
\node[round, below left = .4cm and .8cm of A] (A1) {};
\node[round, below right = .4cm and .8cm  of A] (A2) {};
\node[round, below left = .2cm and .4cm of A1] (A11) {};
\node[alpha, below right = .2cm and .4cm  of A1] (A12) {\small $T_i$};
\node[alpha, below left = .2cm and .4cm of A2] (A21) {\small $F_i$};
\node[round, below right = .2cm and .4cm  of A2] (A22) {};

\node[alpha, below left = .2cm and .4cm of A11] (A111) {\small $F_i$};
\node[alpha, below right = .2cm and .4cm  of A11] (A112) {\small $F_i$};
\node[alpha, below left = .2cm and .4cm of A22] (A221) {\small $T_i$};
\node[alpha, below right = .2cm and .4cm  of A22] (A222) {\small $T_i$};

\draw (A) -- (A1);
\draw (A) -- (A2);
\draw (A1) -- (A11);
\draw (A1) -- (A12);
\draw (A2) -- (A21);
\draw (A2) -- (A22);
\draw (A11) -- (A111);
\draw (A11) -- (A112);
\draw (A22) -- (A221);
\draw (A22) -- (A222);
\end{tikzpicture}}
\caption{The cases $Q_i = \forall$ (left) and $Q_i = \exists$ (right).} \label{fig-exists}
\end{figure}
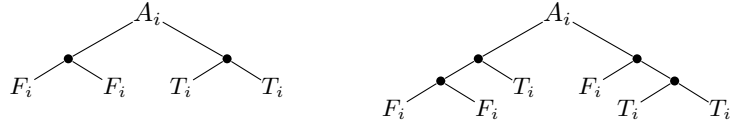
The trees produced by the productions in \eqref{rule-A_i} are shown in Figure~\ref{fig-exists}.
Note that these trees realize the Strahler algebra expressions in \eqref{eq-conjunction} and \eqref{eq-disjunction}
(replace $x$ and $y$ in \eqref{eq-conjunction} and \eqref{eq-disjunction} by $F_i$ and $T_i$, respectively)
and therefore  implement conjunction (for the case $Q_i = \forall$) and disjunction (for the case $Q_i = \exists$).
The nonterminal $F_i$ (resp., $T_i$) stands for setting the Boolean variable $x_i$ to false (resp., true).

In addition, $G$ contains productions for the Boolean formula $\phi(x_1, \ldots, x_n)$. 
In the following, we identify the formula $\phi$ with its syntax tree. This is a binary tree, where all internal nodes are labelled 
with $\wedge$ or $\vee$ and all leaves are labelled with a literal (a variable $x_i$ or a negated variable $\neg x_i$).
Moreover, we can assume w.l.o.g.~that every path from the root to a leaf in $\phi$ has the same length $h$ (the height of $\phi$).  
Every node $A$ of $\phi$ is a nonterminal of $G$.
We identify the nonterminal $A_{n+1}$ with the root node of $\phi$.
For every internal node $A$ of $\phi$ with left child $B$ and right child  $C$ we introduce the following production:
\[ A  \to 
\begin{cases}
[BB] [CC]                                                      & \text{ if $A$ is labelled with $\wedge$}  \\[1mm]
\big[ [BB] C \big] \big[ B [CC] \big]  & \text{ if $A$ is labelled with $\vee$} 
\end{cases} 
\]
Again, the trees corresponding to the right-hand sides of these productions realize in the Strahler algebra the operations
from  \eqref{eq-conjunction} and \eqref{eq-disjunction} and therefore implement conjunction and disjunction.

Finally, for every leaf node $A$ of $\phi$ we introduce the following productions:
\begin{equation} \label{rule-A}
A  \to 
\begin{cases}
\big[ [EE][EE] \big] F_i \mid  [EE] T_i  & \text{ if $A$ is labelled in $\phi$ with $x_i$}  \\[1mm]
\big[ [EE][EE] \big] T_i \mid  [EE] F_i   & \text{ if $A$ is labelled in $\phi$ with $\neg x_i$} 
\end{cases} 
\end{equation}
Note that the Strahler number of the trees corresponding to $\big[ [EE][EE] \big] F_i$ and 
$\big[ [EE][EE] \big] T_i$ (resp.,  $[EE] T_i$ and $[EE] F_i$) is $2$ (resp., $1$).

Let us now analyze the Strahler number of acyclic derivation trees of $G$. 
First of all, there is a unique acyclic derivation tree for $G$, where the productions $T_i \to \varepsilon$ and $F_i \to \varepsilon$ are used
only for those occurrences of $T_i$ and $F_i$ that are produced with the productions from \eqref{rule-A}. In other words, if an 
occurrence of $T_i$ (resp., $F_i$) is produced with \eqref{rule-A_i}, then this occurrence is expanded using the production
$T_i \to A_{i+1} E$ (resp., $F_i \to A_{i+1} E$).  Let us call this acyclic derivation tree the \emph{big tree} $t_{\text{big}}$. 
All other acyclic derivation trees are called
\emph{small}. They can be obtained from $t_{\text{big}}$ by removing from some internal nodes $v$ that are labelled with $T_i$ or $F_i$
all descendants below $v$.

It suffices to show the following:
\begin{enumerate}[(A)]
\item If $\psi$ is true then $\hs(t_{\text{big}}) = 2h+2n+2$.
\item If $\psi$ is false then $\hs(t_{\text{big}}) = 2h+2n+1$.
\end{enumerate}
We can then set $k = 2h+2n+2$. Note that if $\hs(t_{\text{big}}) = 2h+2n+1$ then also every small 
acyclic derivation tree has Strahler number at most $2h+2n+1$
 (removing subtrees from $t_{\text{big}}$ cannot
make the Strahler number larger).

To show the statements (A) and (B), let us first consider an $A_{n+1}$-labelled node $u$ of $t_{\text{big}}$ (recall that the nonterminal
$A_{n+1}$ is also the root of the Boolean formula $\phi$). Then $u$ determines a truth assignment 
$\eta_u : \{x_1, \ldots, x_n\} \to \{0,1\}$ as follows:
$\eta_u(x_i) = 1$ if the nonterminal $T_i$ occurs on the path from the root of $t_{\text{big}}$ to $u$ and
$\eta_u(x_i) = 0$ if the nonterminal $F_i$ occurs on the path from the root of $t_{\text{big}}$ to $u$.
 We can extend this mapping $\eta_u$ to any node $A$ of $\phi$ by evaluating the  subexpression rooted in $A$ under
the truth assignment $\eta_u$. For $\eta_u(A_{n+1})$ we write $\eta_u(\phi)$.

 \begin{claim} \label{claim-eta_t}
 Let  $u$ be an $A_{n+1}$-labelled node of $t_{\text{big}}$, let $A$ be a node of $\phi$ of height $d$ in $\phi$
 and let $v$ be an $A$-labelled node in the subtree $t_{\text{big}}(u)$. Then the following holds:
\begin{enumerate}[(i)]
\item If $\eta_u(A) = 0$ then $\hs(t_{\text{big}}(v)) = 2d+1$.
\item If $\eta_u(A) = 1$ then $\hs(t_{\text{big}}(v)) = 2d+2$.
\end{enumerate} 
\end{claim}

\begin{claimproof}
We show the statement by induction on $d$. If $d=0$ then $A$ is a leaf of $\phi$ that is labelled with some 
literal $x_i$ or $\neg x_i$. If $A$ is labelled with $x_i$ then one of the two productions 
$A \to \big[ [EE][EE] \big] F_i$ or $A \to  [EE] T_i$ is applied in the node $v$ of  $t_{\text{big}}$.
If $\eta_u(A) = \eta_u(x_i) = 1$ then $T_i$ appears on the path from the root to $u$. Since $v$ is a descendant
of $u$ and $t_{\text{big}}$ is acyclic, the production $A \to \big[ [EE][EE] \big] F_i$ must be applied in $v$,
which implies that $\hs(t_{\text{big}}(v)) = 2$. Similarly, if 
$\eta_u(A) =  \eta_u(x_i) = 0$ then the production $A \to  [EE] T_i$ must be applied in $v$ and hence $\hs(t_{\text{big}}(v)) = 1$.
For the case that $A$ is labelled with $\neg x_i$ one can argue analogously.

Assume now that $A$ is labelled with $\wedge$ and let $B$ and $C$ be the two children of $A$ in the Boolean formula $\phi$.
Then, $d \geq 1$ and $B$ and $C$ have height $d-1$ in $\phi$. In the node $v$ of $t_{\text{big}}$ the production 
$A \to [BB] [CC]$ is applied. Let $v_1$ and $v_2$ be the two $B$-labelled grandchildren of $v$ and $v_3$ and $v_4$ be the 
two $C$-labelled grandchildren of $v$ that are produced by $A \to [BB] [CC]$. By induction we get: 
\begin{itemize}
\item If $\eta_u(B) = 0$ then $\hs(t_{\text{big}}(v_1)) = \hs(t_{\text{big}}(v_2)) = 2d-1$.
\item If $\eta_u(B) = 1$ then $\hs(t_{\text{big}}(v_1)) = \hs(t_{\text{big}}(v_2)) = 2d$.
\item If $\eta_u(C) = 0$ then $\hs(t_{\text{big}}(v_3)) = \hs(t_{\text{big}}(v_4)) = 2d-1$.
\item If $\eta_u(C) = 1$ then $\hs(t_{\text{big}}(v_3)) = \hs(t_{\text{big}}(v_4)) = 2d$.
\end{itemize} 
Then the calculations from \eqref{eq-wedge} show (i) and (ii) from the claim.
For the case that $A$ is labelled with $\vee$ we can argue similarly, using \eqref{eq-vee}.
\end{claimproof}
Consider now an $A_i$-labelled node $u$ of $t_{\text{big}}$ for some $1 \leq i \leq n+1$. It determines a partial truth assignment
$\eta_u : \{ x_1, \ldots, x_{i-1} \} \to \{0,1\}$ by the path from the root to $u$: if $T_j$ (resp., $F_j$) lies on this path
for some $1 \le j \le i-1$ 
 then $\eta_u(x_j) = 1$ 
(resp., $\eta_u(x_j) = 0$). Since the subformula $Q_i x_i Q_{i+1} x_{i+1} \cdots Q_n x_n \phi$ of $\psi$ has the free variables
$x_1, \ldots, x_{i-1}$ we can define the truth value $\eta_u(Q_i x_i Q_{i+1} x_{i+1} \cdots Q_n x_n \phi)$ in the obvious way.
Recall that $h$ is the height of $\phi$.

 \begin{claim} \label{claim-quantifiers}
Let $u$ be an $A_i$-labelled node of $t_{\text{big}}$ for some $1 \leq i \leq n+1$. 
Then the following holds:
\begin{itemize}
\item If $\eta_u(Q_i x_i Q_{i+1} x_{i+1} \cdots Q_n x_n \phi) = 1$ then $\hs(t_{\text{big}}(u)) = 2h+2(n-i+1) + 2$.
\item If $\eta_u(Q_i x_i Q_{i+1} x_{i+1} \cdots Q_n x_n \phi) = 0$ then $\hs(t_{\text{big}}(u)) = 2h+2(n-i+1) + 1$.
\end{itemize}
\end{claim}
\begin{figure}[t]
 \tikzset{alpha/.style={inner sep = 0pt, fill=white}}
  \tikzset{round/.style={inner sep = 1.2pt, circle, fill=black}}
 \tikzset{%
  dots/.style args={#1per #2}{%
    line cap=round,
    dash pattern=on 0 off #2/#1
  }
}
\centering{
\begin{tikzpicture}
\node[alpha, label=right:{\small $\!\! u$}] (A) {$A_i$};
\node[round, below left = .6cm and 1.2cm of A] (A1) {};
\node[round, below right = .6cm and 1.2cm  of A] (A2) {};
\node[alpha, below left = .3cm and .6cm of A1, label=right:{\small $\!\! u_1$}] (A11) {\small $F_i$};
\node[alpha, below right = .3cm and .6cm  of A1, label=right:{\small $\!\! u_2$}] (A12) {\small $F_i$};
\node[alpha, below left = .3cm and .6cm of A2, label=right:{\small $\!\! u_3$}] (A21) {\small $T_i$};
\node[alpha, below right = .3cm and .6cm  of A2, label=right:{\small $\!\! u_4$}] (A22) {\small $T_i$};

\node[alpha, below left = .5cm and -.3cm of A11, label=right:{\small $\!\! v_1$}] (A111) {\small $A_{i+1}$};
\node[alpha, below right = .15cm and .3cm  of A11] (A112) {\small $E$};
\node[alpha, below left = .5cm and -.3cm of A12, label=right:{\small $\!\! v_2$}] (A121) {\small $A_{i+1}$};
\node[alpha, below right = .15cm and .3cm  of A12] (A122) {\small $E$};

\node[alpha, below left = .5cm and -.3cm of A21, label=right:{\small $\!\! v_3$}] (A211) {\small $A_{i+1}$};
\node[alpha, below right = .15cm and .3cm  of A21] (A212) {\small $E$};
\node[alpha, below left = .5cm and -.3cm of A22, label=right:{\small $\!\! v_4$}] (A221) {\small $A_{i+1}$};
\node[alpha, below right = .15cm and .3cm  of A22] (A222) {\small $E$};

\node[alpha, above right = .2cm and .4cm of A] (xi-1) {$X_{i-1}$};
\node[alpha, above  = .8cm of xi-1] (xi-2) {$X_{i-2}$};
\node[alpha, above  = 1.2cm of xi-2] (x2) {$X_2$};
\node[alpha, above  = .8cm of x2] (x1) {$X_1$};
\node[alpha, above  = .8cm of x1] (r) {$A_1$};

\draw[decorate,decoration=snake] (r) -- (x1);
\draw[decorate,decoration=snake] (x1) -- (x2);
\draw[decorate,decoration=snake, dots=30 per 1cm] (x2) -- (xi-2);
\draw[decorate,decoration=snake] (xi-2) -- (xi-1);
\draw (xi-1) -- (A);
\draw (A) -- (A1);
\draw (A) -- (A2);
\draw (A1) -- (A11);
\draw (A1) -- (A12);
\draw (A2) -- (A21);
\draw (A2) -- (A22);

\draw (A11) -- (A111);
\draw (A11) -- (A112);
\draw (A12) -- (A121);
\draw (A12) -- (A122);

\draw (A21) -- (A211);
\draw (A21) -- (A212);
\draw (A22) -- (A221);
\draw (A22) -- (A222);
\end{tikzpicture}}
\caption{The case $Q_i = \forall$ in the proof of Claim~\ref{claim-quantifiers}. 
Every $X_j$ is either $T_j$ or $F_j$; 
 $u, u_1$, $u_2$, $u_3$, $u_4$, $v_1$, $v_2$, $v_3$ and $v_4$ are the node names used in the proof of Claim~\ref{claim-quantifiers}.} \label{fig-claim-forall}
\end{figure}
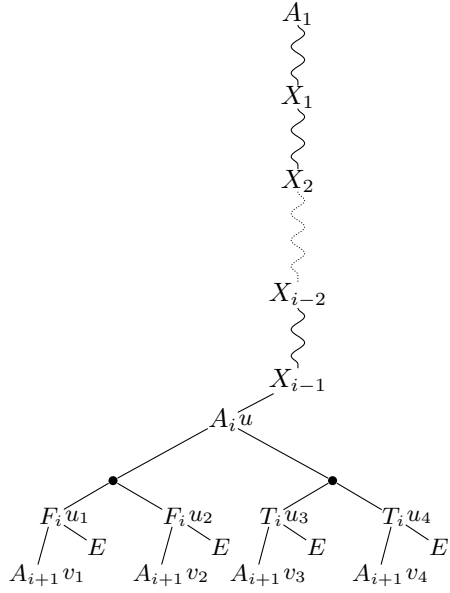
\begin{claimproof}
We prove the statement by induction on $n-i+1$ starting with $i = n+1$ and ending with $i = 1$.
If $i = n+1$ then the statement follows from Claim~\ref{claim-eta_t} for $u=v$ and $d=h$.
Now assume that $1 \le i \le n$. If $Q_i = \forall$ then the production 
$A_i \to [F_i F_i ] [T_i T_i]$ is applied in $u$. 
Let $u_1$ and $u_2$ be the two $F_i$-labelled grandchildren of $u$ and $u_3$ and $u_4$ be the 
two $T_i$-labelled grandchildren of $u$ that are produced by $A_i \to [F_i F_i ] [T_i T_i]$. 
Moreover, let $v_j$ be the left child of $u_j$; it is labelled with $A_{i+1}$.

We obtain the pattern shown in Figure~\ref{fig-claim-forall}.
If $\eta_u(\forall_i x_i Q_{i+1} x_{i+1} \cdots Q_n x_n \phi) = 1$ then 
$\eta_{v_j}(Q_{i+1} x_{i+1} \cdots Q_n x_n \phi) = 1$ for all $1 \le j \le 4$. By induction we have
$\hs(t_{\text{big}}(u_j)) = \hs(t_{\text{big}}(v_j)) = 2h+2(n-i) + 2$ for all $1 \le j \le 4$, which yields
$\hs(t_{\text{big}}(u)) = 2h+2(n-i) + 4 = 2h+2(n-i+1) + 2$.
On the other hand, if $\eta_u(\forall_i x_i Q_{i+1} x_{i+1} \cdots Q_n x_n \phi) = 0$ then 
one of the following two cases holds: 
\begin{itemize}
\item $\eta_{v_1}(Q_{i+1} x_{i+1} \cdots Q_n x_n \phi) = \eta_{v_2}(Q_{i+1} x_{i+1} \cdots Q_n x_n \phi) = 0$,
\item $\eta_{v_3}(Q_{i+1} x_{i+1} \cdots Q_n x_n \phi) = \eta_{v_4}(Q_{i+1} x_{i+1} \cdots Q_n x_n \phi) = 0$.
\end{itemize}
W.l.o.g.~assume that the first case holds. By induction, we obtain 
$\hs(t_{\text{big}}(u_1)) =  \hs(t_{\text{big}}(u_2)) = 2h+2(n-i) + 1$ and 
$2h+2(n-i) + 1 \le \hs(t_{\text{big}}(u_3)) = \hs(t_{\text{big}}(u_4)) \le 2h+2(n-i) + 2$.
This implies  
$\hs(t_{\text{big}}(u)) = 2h+2(n-i) + 3 = 2h+2(n-i+1) + 1$. The case where $Q_i = \exists$ can be dealt with a similar
reasoning.
\end{claimproof} 
Claim~\ref{claim-quantifiers} yields for $i=1$ the above statements (A) and (B), which concludes the $\PSPACE$-hardness proof.
\end{proof}
We conclude the paper with some remarks on the 
$\P$-hardness of $\CNF^{\ge k}$ (Theorem~\ref{thm-CNF}(i)). Our $\P$-hardness proof uses
a trivial reduction from the emptiness problem for CNF-grammars.
One can eliminate this source of hardness by adding to the input a certificate, ensuring that $G$ is \emph{productive}. 
A \emph{p-certificate} for a CNF-grammar $G = (N,S,P)$ is a function $f : N \to N^*$ such that
(i) $(A \to f(A)) \in P$ for every $A \in N$ and 
(ii) the directed graph $(N, \{ (A,B) \in N \times N : \text{ $B$ occurs in $f(A)$} \})$ is acyclic. It is easy to see that
there is a p-certificate for $G$ if and only if $G$ is productive.

We define the problems $\rCNF^\ge$ and $\rCNF^{\ge k}$ in the same way as $\CNF^\ge$ and $\CNF^{\ge k}$, respectively, except 
that the input also contains a p-certificate for $G$. 

\begin{theorem} \label{theorem-rCNF-Ptime}
$\rCNF^\ge$ is $\P$-complete and 
$\rCNF^{\ge k}$ is $\NL$-complete for every $k \geq 2$.
\end{theorem}

\begin{proof}
Membership of $\rCNF^{\ge}$ in $\P$ follows directly from Theorem~\ref{thm-CNF}(i) and 
$\P$-hardness of $\rCNF^\ge$ is a consequence of Theorem~\ref{thm-compressed}(i), since $\StNdag$ is a restriction of $\rCNF^\ge$. 
A binary DAG $\mathcal{D} = (V,v_0,\gamma)$ without a node-labelling function $\lambda$
can be identified with the CNF-grammar $G = (V, v_0, \{ (v \to \gamma(v)) : v \in V\})$
 that produces the single derivation tree $\unfold(\mathcal{D})$. The function $\gamma$ is then a p-certificate.

The $\NL$ upper bound for $\rCNF^{\ge k}$ can be obtained similarly 
 to the proof of Theorem~\ref{thm-compressed}(iii): Let $G = (N,S,P)$ be the input CNF grammar together with a p-certificate.
The algorithm stores an active statement $\mathcal{S}_a = [\hs_A \geq m]$ for $A \in N$ and $0 \le k \le m$. 
If $m=0$ then the algorithm accepts (which is correct since $A$ is the root of a derivation tree).
If $m > 0$ then the algorithm rejects if $A \to \varepsilon$ is the only production for $A$. Otherwise, it
guesses existentially a production $(A \to A_1 A_2) \in P$ and branches
existentially to one of the following cases:
 \begin{itemize}
\item set $\mathcal{S}_a = [\hs_{A_1} \geq m]$ and continue,
\item set $\mathcal{S}_a = [\hs_{A_2} \geq m]$ and continue,
\item universally guess an $i \in \{1,2\}$ and set $\mathcal{S}_a = [\hs_{A_i} \geq m-1]$.
\end{itemize}
The number of alternations is again bounded by the fixed constant $k$.

The $\NL$ lower bound for $\rCNF^{\ge k}$ with $k \geq 2$
can be shown by a reduction from the graph accessibility problem. 
Consider a binary DAG $\mathcal{D} = (V, v_0, \gamma)$ and a target leaf $v_t$
as in the $\NL$-hardness proof of 
Theorem~\ref{thm-compressed}(iii). We define a CNF grammar $G = (V \uplus \{A,B\}, v_0, P)$
with the following productions:
\begin{itemize}
\item $v \to \varepsilon$ if $v \in V \setminus \{v_t\}$ and $\gamma(v) = \varepsilon$,
\item $v_t \to BB$, $B \to AA$ and $A \to \varepsilon$.
\item $v \to v_1 A$ and $v \to v_2 A$ if $\gamma(v) = v_1 v_2$,
\end{itemize}
If there is a path from $v_0$ to $v_t$ then $G$ has a derivation tree with Strahler number $2$ and
if there is no path from $v_0$ to $v_t$ then all derivation trees of $G$ have Strahler number at most one.
A p-certificate for $G$ can be obtained by mapping $v$ to $v_1 A$ in case $\gamma(v) = v_1 v_2$.
\end{proof}
Table~\ref{table-CNF} summarizes the results from this section.

\begin{table}[t]
\begin{center}
    \begin{tabular}{c|c|c|c}
             &  $\CNF$             & $\rCNF$            & $\acCNF$ \\ \hline
$(\cdot)^\ge$    & $\P$-comp. & $\P$-comp. & $\PSPACE$-comp.   \\ \hline
$(\cdot)^{\ge k}$ with $k \geq 1$ & $\P$-comp.  & $\NL$-comp.  ($k \geq 2$)  & $\NP$-comp. ($k \geq 2$) \\
    \end{tabular}
  \end{center}
\caption{\label{table-CNF} Complexity results for the problems from Section~\ref{sec-CNF}}
\end{table}

\section{Open problems}

We conclude the paper with some open problems.

\subparagraph{Computing Strahler numbers for unranked trees.}
Strahler numbers have been also defined for unranked trees (trees, where nodes can have any number of children):
Consider an unranked tree $t$, where 
$t_1, \ldots, t_k$ ($k \geq 0$) are the subtrees rooted in the children of the root of $t$.
We define the Strahler number $\hs(t)$ of $t$ inductively
as follows: if $k=0$ then $\hs(t) = 0$. If $k \geq 1$ then let $n_i = \hs(t_i)$ and $n = \max\{n_1, \ldots, n_k\}$.
If $n$ has a unique occurrence in the list $n_1, \ldots, n_k$, then $\hs(t) = n$, otherwise $\hs(t) = n+1$.
We conjecture that our $\NC^1$-algorithm for computing the Strahler number of a binary tree can be adapted
to unranked trees, but this seems to require new ideas.
Tree straight-line programs could be replaced by forest straight-line programs \cite{GasconLMRS20} that work for unranked trees.
For this one has to prove a variant of Theorem~\ref{theorem tree->tslp} for forest straight-line programs. In addition, one needs
a variant of Lemma~\ref{lemma-composition} for unranked trees, which is not obvious.

\subparagraph{Expression evaluation for the max-plus semiring.}
Closely related but slightly different to the computation of Strahler numbers 
is the problem of evaluating expressions over the max-plus semiring $(\mathbb{N},\max,+)$.
If the input numbers are given in unary encoding, then the problem is logspace reducible to the evaluation of arithmetic expressions over $(\mathbb{N},+,\times)$
and hence belongs to $\L$~\cite{JakobyT06}.
The complexity is open if the input numbers are encoded in binary.
Let us also mention that the longest (or shortest~\cite[Lemma~3.3]{JakobyT06}) path problem in a directed graph with binary encoded weights is in $\AC^1$,
using matrix powering over the max-plus semiring,
but it is a longstanding open problem whether it lies in a smaller complexity class~\cite[p.~13]{Cook85}, see also \cite{GanardiSZ24,ShakibaBZ25} for recent applications.

\subparagraph{Computing path width of trees.}
A problem that is related to the computation of the Strahler number of a tree $t$ is the computation of the path width of an undirected
tree. Let us define $\hs(t)$ for an undirected tree $t$ as the minimal Strahler number of a rooted tree that can be obtained
by choosing a root in $t$. Then the following relationship is stated in \cite{EsparzaLS14}:
$\mathsf{pathwidth}(t)-1 \le \hs(t) \le 2\cdot\mathsf{pathwidth}(t)$.
It is shown in \cite{Scheffler1990} that the path width of an undirected tree $t$ can be computed in linear time. It would 
be interesting to know whether it can be computed in logspace.

\def\cprime{$'$} \def\cprime{$'$}


\begin{thebibliography}{10}

\bibitem{ADDGG04}
Alex {Arenas}, Leon {Danon}, Albert {D{\'\i}az-Guilera}, Pablo~M. {Gleiser},
  and Roger {Guimer{\'a}}.
\newblock Community analysis in social networks.
\newblock {\em European Physical Journal B}, 38(2):373--380, 2004.
\newblock \href {https://doi.org/10.1140/epjb/e2004-00130-1}
  {\path{doi:10.1140/epjb/e2004-00130-1}}.

\bibitem{AroBar09}
Sanjeev Arora and Boaz Barak.
\newblock {\em Computational Complexity - A Modern Approach}.
\newblock Cambridge University Press, 2009.
\newblock \href {https://doi.org/10.1017/CBO9780511804090}
  {\path{doi:10.1017/CBO9780511804090}}.

\bibitem{AtigG11}
Mohamed~Faouzi Atig and Pierre Ganty.
\newblock Approximating {Petri} net reachability along context-free traces.
\newblock In {\em Proceedings of the {IARCS} Annual Conference on Foundations
  of Software Technology and Theoretical Computer Science, {FSTTCS} 2011},
  volume~13 of {\em LIPIcs}, pages 152--163. Schloss Dagstuhl - Leibniz-Zentrum
  f{\"{u}}r Informatik, 2011.
\newblock \href {https://doi.org/10.4230/LIPICS.FSTTCS.2011.152}
  {\path{doi:10.4230/LIPICS.FSTTCS.2011.152}}.

\bibitem{BIS90}
David A.~Mix Barrington, Neil Immerman, and Howard Straubing.
\newblock On uniformity within $\text{NC}^1$.
\newblock {\em Journal of Computer and System Sciences}, 41:274--306, 1990.
\newblock \href {https://doi.org/10.1016/0022-0000(90)90022-D}
  {\path{doi:10.1016/0022-0000(90)90022-D}}.

\bibitem{BeaudryM95}
Martin Beaudry and Pierre McKenzie.
\newblock Circuits, matrices, and nonassociative computation.
\newblock {\em Journal of Computer and System Sciences}, 50(3):441--455, 1995.
\newblock \href {https://doi.org/10.1006/JCSS.1995.1035}
  {\path{doi:10.1006/JCSS.1995.1035}}.

\bibitem{Ben-OrC92}
Michael Ben{-}Or and Richard Cleve.
\newblock Computing algebraic formulas using a constant number of registers.
\newblock {\em {SIAM} Journal on Computing}, 21(1):54--58, 1992.
\newblock \href {https://doi.org/10.1137/0221006} {\path{doi:10.1137/0221006}}.

\bibitem{BiziereC25}
Clotilde Bizi{\`{e}}re and Wojciech Czerwinski.
\newblock Reachability in one-dimensional pushdown vector addition systems is
  decidable.
\newblock In {\em Proceedings of the 57th Annual {ACM} Symposium on Theory of
  Computing, {STOC} 2025}, pages 1851--1862. {ACM}, 2025.
\newblock \href {https://doi.org/10.1145/3717823.3718149}
  {\path{doi:10.1145/3717823.3718149}}.

\bibitem{Brent74}
Richard~P. Brent.
\newblock The parallel evaluation of general arithmetic expressions.
\newblock {\em Journal of the {ACM}}, 21(2):201--206, 1974.
\newblock \href {https://doi.org/10.1145/321812.321815}
  {\path{doi:10.1145/321812.321815}}.

\bibitem{Bus87}
Samuel~R. Buss.
\newblock The {Boolean} formula value problem is in {ALOGTIME}.
\newblock In {\em Proceedings of the 19th {A}nnual {S}ymposium on {T}heory of
  {C}omputing, STOC 1987}, pages 123--131. ACM Press, 1987.
\newblock \href {https://doi.org/10.1145/28395.28409}
  {\path{doi:10.1145/28395.28409}}.

\bibitem{BCGR92}
Samuel~R. Buss, Stephen~A. Cook, A.~Gupta, and V.~Ramachandran.
\newblock An optimal parallel algorithm for formula evaluation.
\newblock {\em SIAM Journal on Computing}, 21(4):755--780, 1992.
\newblock \href {https://doi.org/10.1137/0221046} {\path{doi:10.1137/0221046}}.

\bibitem{ChiuDL01}
Andrew Chiu, George~I. Davida, and Bruce~E. Litow.
\newblock Division in logspace-uniform \emph{NC}\({}^{\mbox{1}}\).
\newblock {\em RAIRO -- Theoretical Informatics and Applications},
  35(3):259--275, 2001.
\newblock \href {https://doi.org/10.1051/ITA:2001119}
  {\path{doi:10.1051/ITA:2001119}}.

\bibitem{ChMo90}
Michal~P. Chytil and Burkhard Monien.
\newblock Caterpillars and context-free languages.
\newblock In {\em Proceedings of the 7th Annual Symposium on Theoretical
  Aspects of Computer Science, STACS 1990}, volume 415 of {\em Lecture Notes in
  Computer Science}, pages 70--81. Springer, 1990.
\newblock \href {https://doi.org/10.1007/3-540-52282-4}
  {\path{doi:10.1007/3-540-52282-4}}.

\bibitem{CookM24}
James Cook and Ian Mertz.
\newblock Tree evaluation is in space ${O}(\log n \cdot \log \log n)$.
\newblock In {\em Proceedings of the 56th Annual {ACM} Symposium on Theory of
  Computing, {STOC} 2024}, pages 1268--1278. {ACM}, 2024.
\newblock \href {https://doi.org/10.1145/3618260.3649664}
  {\path{doi:10.1145/3618260.3649664}}.

\bibitem{Cook85}
Stephen~A. Cook.
\newblock A taxonomy of problems with fast parallel algorithms.
\newblock {\em Information and Control}, 64(1-3):2--21, 1985.
\newblock \href {https://doi.org/10.1016/S0019-9958(85)80041-3}
  {\path{doi:10.1016/S0019-9958(85)80041-3}}.

\bibitem{Courcelle90}
Bruno Courcelle.
\newblock Graph rewriting: An algebraic and logic approach.
\newblock In Jan van Leeuwen, editor, {\em Handbook of Theoretical Computer
  Science, Volume {B:} Formal Models and Semantics}, pages 193--242. Elsevier
  and {MIT} Press, 1990.
\newblock \href {https://doi.org/10.1016/B978-0-444-88074-1.50010-X}
  {\path{doi:10.1016/B978-0-444-88074-1.50010-X}}.

\bibitem{DaJuTh20}
Laure Daviaud, Marcin Jurdzi\'{n}ski, and K.~S. Thejaswini.
\newblock The {Strahler} number of a parity game.
\newblock In {\em Proceedings of the 47th International Colloquium on Automata,
  Languages, and Programming, ICALP 2020}, volume 168 of {\em LIPIcs}, pages
  123:1--123:19. Schloss Dagstuhl -- Leibniz-Zentrum f{\"u}r Informatik, 2020.
\newblock \href {https://doi.org/10.4230/LIPIcs.ICALP.2020.123}
  {\path{doi:10.4230/LIPIcs.ICALP.2020.123}}.

\bibitem{DevKru95}
Luc Devroye and Paul Kruszewski.
\newblock A note on the {Horton-Strahler} number for random trees.
\newblock {\em Information Processing Letters}, 56(2):95--99, 1995.
\newblock \href {https://doi.org/10.1016/0020-0190(95)00114-R}
  {\path{doi:10.1016/0020-0190(95)00114-R}}.

\bibitem{DeKru96}
Luc Devroye and Paul Kruszewski.
\newblock On the {Horton-Strahler} number for random tries.
\newblock {\em RAIRO -- Theoretical Informatics and Applications},
  30(5):443--456, 1996.
\newblock \href {https://doi.org/10.1051/ita/1996300504431}
  {\path{doi:10.1051/ita/1996300504431}}.

\bibitem{DowneyST80}
Peter~J. Downey, Ravi Sethi, and Robert~Endre Tarjan.
\newblock Variations on the common subexpression problem.
\newblock {\em Journal of the {ACM}}, 27(4):758--771, 1980.
\newblock \href {https://doi.org/10.1145/322217.322228}
  {\path{doi:10.1145/322217.322228}}.

\bibitem{EhRoVe81}
Andrzej Ehrenfeucht, Grzegorz Rozenberg, and Dirk Vermeir.
\newblock On {ETOL} systems with finite tree-rank.
\newblock {\em SIAM Journal on Computing}, 10(1):40--58, 1981.
\newblock \href {https://doi.org/10.1137/0210004} {\path{doi:10.1137/0210004}}.

\bibitem{ElberfeldJT10}
Michael Elberfeld, Andreas Jakoby, and Till Tantau.
\newblock Logspace versions of the theorems of {Bodlaender} and {Courcelle}.
\newblock In {\em Proceedings of the 51th Annual {IEEE} Symposium on
  Foundations of Computer Science, {FOCS} 2010}, pages 143--152. {IEEE}
  Computer Society, 2010.
\newblock \href {https://doi.org/10.1109/FOCS.2010.21}
  {\path{doi:10.1109/FOCS.2010.21}}.

\bibitem{ElberfeldJT12}
Michael Elberfeld, Andreas Jakoby, and Till Tantau.
\newblock Algorithmic meta theorems for circuit classes of constant and
  logarithmic depth.
\newblock In {\em Proceedings of the 29th International Symposium on
  Theoretical Aspects of Computer Science, {STACS} 2012}, volume~14 of {\em
  LIPIcs}, pages 66--77. Schloss Dagstuhl - Leibniz-Zentrum f{\"{u}}r
  Informatik, 2012.
\newblock \href {https://doi.org/10.4230/LIPICS.STACS.2012.66}
  {\path{doi:10.4230/LIPICS.STACS.2012.66}}.

\bibitem{Ers58}
Andrey~Petrovich Ershov.
\newblock On programming of arithmetic operations.
\newblock {\em Communications of the ACM}, 1(8):3--6, 1958.
\newblock \href {https://doi.org/10.1145/368892.368907}
  {\path{doi:10.1145/368892.368907}}.

\bibitem{EKL10}
Javier Esparza, Stefan Kiefer, and Michael Luttenberger.
\newblock Newtonian program analysis.
\newblock {\em Journal of the {ACM}}, 57(6), 2010.
\newblock \href {https://doi.org/10.1145/1857914.1857917}
  {\path{doi:10.1145/1857914.1857917}}.

\bibitem{EsparzaLS14}
Javier Esparza, Michael Luttenberger, and Maximilian Schlund.
\newblock A brief history of {Strahler} numbers.
\newblock In {\em Proceedings of the 8th International Conference on Language
  and Automata Theory and Applications, {LATA} 2014}, volume 8370 of {\em
  Lecture Notes in Computer Science}, pages 1--13. Springer, 2014.
\newblock \href {https://doi.org/10.1007/978-3-319-04921-2\_1}
  {\path{doi:10.1007/978-3-319-04921-2\_1}}.

\bibitem{Etessami97}
Kousha Etessami.
\newblock Counting quantifiers, successor relations, and logarithmic space.
\newblock {\em Journal of Computer and System Sciences}, 54(3):400--411, 1997.
\newblock \href {https://doi.org/10.1006/JCSS.1997.1485}
  {\path{doi:10.1006/JCSS.1997.1485}}.

\bibitem{FRV79}
Philippe Flajolet, Jean{-}Claude Raoult, and Jean Vuillemin.
\newblock The number of registers required for evaluating arithmetic
  expressions.
\newblock {\em Theoretical Computer Science}, 9(1):99--125, 1979.
\newblock \href {https://doi.org/10.1016/0304-3975(79)90009-4}
  {\path{doi:10.1016/0304-3975(79)90009-4}}.

\bibitem{FlajoletSS90}
Philippe Flajolet, Paolo Sipala, and Jean{-}Marc Steyaert.
\newblock Analytic variations on the common subexpression problem.
\newblock In {\em Proceedings of the 17th International Colloquium on Automata,
  Languages and Programming, ICALP 1990}, volume 443 of {\em Lecture Notes in
  Computer Science}, pages 220--234. Springer, 1990.
\newblock \href {https://doi.org/10.1007/BFB0032034}
  {\path{doi:10.1007/BFB0032034}}.

\bibitem{GanardiL19}
Moses Ganardi and Markus Lohrey.
\newblock A universal tree balancing theorem.
\newblock {\em {ACM} Transactions on Computation Theory}, 11(1):1:1--1:25,
  2019.
\newblock \href {https://doi.org/10.1145/3278158} {\path{doi:10.1145/3278158}}.

\bibitem{GanardiL26}
Moses Ganardi and Markus Lohrey.
\newblock On the complexity of computing strahler numbers.
\newblock In {\em Proceedings of the 43rd International Symposium on
  Theoretical Aspects of Computer Science, {STACS} 2026}, LIPIcs, pages
  41:1--41:22. Schloss Dagstuhl - Leibniz-Zentrum f{\"{u}}r Informatik, 2026.
\newblock \href {https://doi.org/10.4230/LIPICS.STACS.2026.41}
  {\path{doi:10.4230/LIPICS.STACS.2026.41}}.


\bibitem{GanardiSZ24}
Moses Ganardi, Irmak Saglam, and Georg Zetzsche.
\newblock Directed regular and context-free languages.
\newblock In {\em Proceedings of the 41st International Symposium on
  Theoretical Aspects of Computer Science, {STACS} 2024}, volume 289 of {\em
  LIPIcs}, pages 36:1--36:20. Schloss Dagstuhl - Leibniz-Zentrum f{\"{u}}r
  Informatik, 2024.
\newblock \href {https://doi.org/10.4230/LIPICS.STACS.2024.36}
  {\path{doi:10.4230/LIPICS.STACS.2024.36}}.

\bibitem{GarJoh79}
Michael~R. Garey and David~S. Johnson.
\newblock {\em Computers and Intractability: A Guide to the Theory of
  {NP}--completeness}.
\newblock Freeman, 1979.

\bibitem{GasconLMRS20}
Adri{\`{a}} Gasc{\'{o}}n, Markus Lohrey, Sebastian Maneth, Carl~Philipp Reh,
  and Kurt Sieber.
\newblock Grammar-based compression of unranked trees.
\newblock {\em Theory of Computing Systems}, 64(1):141--176, 2020.
\newblock \href {https://doi.org/10.1007/S00224-019-09942-Y}
  {\path{doi:10.1007/S00224-019-09942-Y}}.

\bibitem{GiSpa68}
Seymour Ginsburg and Edwin~H. Spanier.
\newblock Derivation-bounded languages.
\newblock {\em Journal of Computer and System Sciences}, 2(3):228--250, 1968.
\newblock \href {https://doi.org/10.1016/S0022-0000(68)80009-1}
  {\path{doi:10.1016/S0022-0000(68)80009-1}}.

\bibitem{Gruska71b}
Jozef Gruska.
\newblock A few remarks on the index of context-free grammars and languages.
\newblock {\em Information and Control}, 19(3):216--223, 1971.
\newblock \href {https://doi.org/10.1016/S0019-9958(71)90095-7}
  {\path{doi:10.1016/S0019-9958(71)90095-7}}.

\bibitem{HeAlBa02}
William Hesse, Eric Allender, and David A.~Mix Barrington.
\newblock Uniform constant-depth threshold circuits for division and iterated
  multiplication.
\newblock {\em Journal of Computer and System Sciences}, 65:695--716, 2002.
\newblock \href {https://doi.org/10.1016/S0022-0000(02)00025-9}
  {\path{doi:10.1016/S0022-0000(02)00025-9}}.

\bibitem{HoUl79}
John~E. Hopcroft and Jeffrey~D. Ullman.
\newblock {\em Introduction to Automata Theory, Languages and Computation}.
\newblock Addison--Wesley, Reading, MA, 1979.

\bibitem{Ho45}
Robert~E. Horton.
\newblock Erosional development of streams and their drainage basins:
  hydro-physical approach to quantitative morphology.
\newblock {\em Geological Society of America Bulletin}, 56(3):275--370, 1945.
\newblock \href {https://doi.org/10.1130/0016-7606(1945)56[275:EDOSAT]2.0.CO;2}
  {\path{doi:10.1130/0016-7606(1945)56[275:EDOSAT]2.0.CO;2}}.

\bibitem{Imm88}
Neil Immerman.
\newblock Nondeterministic space is closed under complementation.
\newblock {\em SIAM Journal on Computing}, 17(5):935--938, 1988.
\newblock \href {https://doi.org/10.1137/0217058} {\path{doi:10.1137/0217058}}.

\bibitem{JakobyT06}
Andreas Jakoby and Till Tantau.
\newblock Computing shortest paths in series-parallel graphs in logarithmic
  space.
\newblock In {\em Complexity of Boolean Functions}, volume 06111 of {\em
  Dagstuhl Seminar Proceedings}. Internationales Begegnungs- und
  Forschungszentrum f\"ur Informatik (IBFI), Schloss Dagstuhl, Germany, 2006.
\newblock \href {https://doi.org/10.4230/DagSemProc.06111.6}
  {\path{doi:10.4230/DagSemProc.06111.6}}.

\bibitem{Kemp79}
Rainer Kemp.
\newblock The average number of registers needed to evaluate a binary tree
  optimally.
\newblock {\em Acta Informatica}, 11(4):363--372, 1979.
\newblock \href {https://doi.org/10.1007/BF00289094}
  {\path{doi:10.1007/BF00289094}}.

\bibitem{Kosaraju90}
S.~Rao Kosaraju.
\newblock On parallel evaluation of classes of circuits.
\newblock In {\em Proceedings of the 10th Conference on Foundations of Software
  Technology and Theoretical Computer Science}, volume 472 of {\em Lecture
  Notes in Computer Science}, pages 232--237. Springer, 1990.
\newblock \href {https://doi.org/10.1007/3-540-53487-3\_48}
  {\path{doi:10.1007/3-540-53487-3\_48}}.

\bibitem{KrebsLL17}
Andreas Krebs, Nutan Limaye, and Michael Ludwig.
\newblock A unified method for placing problems in polylogarithmic depth.
\newblock In {\em Proceedings of the 37th {IARCS} Annual Conference on
  Foundations of Software Technology and Theoretical Computer Science, {FSTTCS}
  2017}, volume~93 of {\em LIPIcs}, pages 36:36--36:15. Schloss Dagstuhl -
  Leibniz-Zentrum f{\"{u}}r Informatik, 2017.
\newblock \href {https://doi.org/10.4230/LIPICS.FSTTCS.2017.36}
  {\path{doi:10.4230/LIPICS.FSTTCS.2017.36}}.

\bibitem{Kru99}
Paul Kruszewski.
\newblock A note on the {Horton-Strahler} number for random binary search
  trees.
\newblock {\em Information Processing Letters}, 69(1):47--51, 1999.
\newblock \href {https://doi.org/https://doi.org/10.1016/S0020-0190(98)00192-6}
  {\path{doi:https://doi.org/10.1016/S0020-0190(98)00192-6}}.

\bibitem{Loh01rta}
Markus Lohrey.
\newblock On the parallel complexity of tree automata.
\newblock In {\em Proceedings of the {12th} International Conference on Rewrite
  Techniques and Applications, RTA 2001}, number 2051 in Lecture Notes in
  Computer Science, pages 201--215. Springer, 2001.
\newblock \href {https://doi.org/10.1007/3-540-45127-7\_16}
  {\path{doi:10.1007/3-540-45127-7\_16}}.

\bibitem{Lohrey15}
Markus Lohrey.
\newblock Grammar-based tree compression.
\newblock In {\em Proceedings of the 19th International Conference on
  Developments in Language Theory, {DLT} 2015}, number 9168 in Lecture Notes in
  Computer Science, pages 46--57. Springer, 2015.
\newblock \href {https://doi.org/10.1007/978-3-319-21500-6\_3}
  {\path{doi:10.1007/978-3-319-21500-6\_3}}.

\bibitem{LRZ25}
Markus Lohrey, Andreas Rosowski, and Georg Zetzsche.
\newblock Membership problems in finite groups.
\newblock {\em Journal of Algebra}, 675:23--58, 2025.
\newblock \href {https://doi.org/10.1016/j.jalgebra.2025.03.011}
  {\path{doi:10.1016/j.jalgebra.2025.03.011}}.

\bibitem{PiSaSo12}
Carine Pivoteau, Bruno Salvy, and Mich{\`e}le Soria.
\newblock Algorithms for combinatorial structures: Well-founded systems and
  {Newton} iterations.
\newblock {\em Journal of Combinatorial Theory, Series A}, 119(8):1711--1773,
  2012.
\newblock \href {https://doi.org/10.1016/j.jcta.2012.05.007}
  {\path{doi:10.1016/j.jcta.2012.05.007}}.

\bibitem{ReinhardtA00}
Klaus Reinhardt and Eric Allender.
\newblock Making nondeterminism unambiguous.
\newblock {\em {SIAM} Journal on Computing}, 29(4):1118--1131, 2000.
\newblock \href {https://doi.org/10.1137/S0097539798339041}
  {\path{doi:10.1137/S0097539798339041}}.

\bibitem{Ruz81}
Walter~L. Ruzzo.
\newblock On uniform circuit complexity.
\newblock {\em Journal of Computer and System Sciences}, 22(3):365 -- 383,
  1981.
\newblock \href {https://doi.org/10.1016/0022-0000(81)90038-6}
  {\path{doi:10.1016/0022-0000(81)90038-6}}.

\bibitem{Scheffler1990}
Petra Scheffler.
\newblock A linear algorithm for the pathwidth of trees.
\newblock In Rainer Bodendiek and Rudolf Henn, editors, {\em Topics in
  Combinatorics and Graph Theory: Essays in Honour of {Gerhard Ringel}}, pages
  613--620. Physica-Verlag HD, 1990.
\newblock \href {https://doi.org/10.1007/978-3-642-46908-4_70}
  {\path{doi:10.1007/978-3-642-46908-4_70}}.

\bibitem{ShakibaBZ25}
Yousef Shakiba, Henry Sinclair{-}Banks, and Georg Zetzsche.
\newblock A complexity dichotomy for semilinear target sets in automata with
  one counter.
 \newblock In {\em Proceedings of the 40th Annual {ACM/IEEE} Symposium on Logic
  in Computer Science, {LICS} 2025}, pages 594--608. {IEEE}, 2025.
\newblock \href {https://doi.org/10.1109/LICS65433.2025.00051}
  {\path{doi:10.1109/LICS65433.2025.00051}}.

\bibitem{SOS23}
Ekaterina Shemetova, Alexander Okhotin, and Semyon Grigorev.
\newblock Rational index of languages defined by grammars with bounded
  dimension of parse trees.
\newblock {\em Theory of Computing Systems}, 68(3):487--511, 2023.
\newblock \href {https://doi.org/10.1007/s00224-023-10159-3}
  {\path{doi:10.1007/s00224-023-10159-3}}.

\bibitem{Spira71}
Philip~M. Spira.
\newblock On time-hardware complexity tradeoffs for boolean functions.
\newblock In {\em Proceedings of the 4th Hawaii Symposium on System Sciences},
  pages 525--527, 1971.

\bibitem{Stra52}
Arthur~N. Strahler.
\newblock Hypsometric (area-altitude) analysis of erosional topology.
\newblock {\em Geological Society of America Bulletin}, 63(11):1117--1142,
  1952.
\newblock \href
  {https://doi.org/10.1130/0016-7606(1952)63[1117:HAAOET]2.0.CO;2}
  {\path{doi:10.1130/0016-7606(1952)63[1117:HAAOET]2.0.CO;2}}.

\bibitem{Sudborough78}
Ivan~Hal Sudborough.
\newblock On the tape complexity of deterministic context-free languages.
\newblock {\em Journal of the {ACM}}, 25(3):405--414, 1978.
\newblock \href {https://doi.org/10.1145/322077.322083}
  {\path{doi:10.1145/322077.322083}}.

\bibitem{SwernofskyW15}
Joseph Swernofsky and Michael Wehar.
\newblock On the complexity of intersecting regular, context-free, and tree
  languages.
\newblock In {\em Proceedings of the 42nd International Colloquium Automata,
  Languages, and Programming, Part {II}, {ICALP} 2015}, volume 9135 of {\em
  Lecture Notes in Computer Science}, pages 414--426. Springer, 2015.
\newblock \href {https://doi.org/10.1007/978-3-662-47666-6\_33}
  {\path{doi:10.1007/978-3-662-47666-6\_33}}.

\bibitem{Vie90}
Xavier Viennot.
\newblock Trees.
\newblock In {\em Mots, m{\'e}langes offert {\'a} M.P.~Sch\"utzenberger}.
  Hermes, Paris, 1990.
\newblock Available online at \url{http://www.xavierviennot.org}.

\bibitem{Vollmer99}
Heribert Vollmer.
\newblock {\em Introduction to Circuit Complexity - {A} Uniform Approach}.
\newblock Texts in Theoretical Computer Science. An {EATCS} Series. Springer,
  1999.
\newblock \href {https://doi.org/10.1007/978-3-662-03927-4}
  {\path{doi:10.1007/978-3-662-03927-4}}.

\bibitem{Williams25}
R.~Ryan Williams.
\newblock Simulating time with square-root space.
\newblock In Michal Kouck{\'{y}} and Nikhil Bansal, editors, {\em Proceedings
  of the 57th Annual {ACM} Symposium on Theory of Computing, {STOC} 2025},
  pages 13--23. {ACM}, 2025.
\newblock \href {https://doi.org/10.1145/3717823.3718225}
  {\path{doi:10.1145/3717823.3718225}}.

\end{thebibliography}
\end{document}